\newif\iffull
\fulltrue

\iffull
  \documentclass[acmsmall,nonacm]{acmart}
\else
  \documentclass[acmsmall]{acmart}
\fi

\usepackage{graphicx}
\usepackage{amsthm}
\usepackage{amsmath}
\usepackage{mathtools}
\usepackage{thmtools}
\usepackage{algorithm}
\usepackage[noend]{algpseudocode}
\makeatletter
\renewcommand{\ALG@beginalgorithmic}{\small}
\makeatother
\usepackage{float}
\usepackage{tikz}
\usetikzlibrary{positioning,shapes,fit,calc}
\usepackage{xcolor}

\definecolor{eBlue}{HTML}{2A78D6}
\definecolor{eRed}{HTML}{E34948}
\definecolor{eAqua}{HTML}{1BAF7A}
\definecolor{eYellow}{HTML}{EDA100}
\definecolor{eViolet}{HTML}{4A3AA7}
\definecolor{eGreen}{HTML}{008300}

\AtBeginDocument{%
  }

\newtheorem{example}{Example}

\newcommand{\VCSpace}{\Psi}
\newcommand{\VCFunc}{\ve{V}}
\newcommand{\VCVec}{\psi}

\newcommand{\vars}[1]{\textsf{vars}(#1)}

\newcommand{\ve}[1]{\mathbf{#1}}
\newcommand{\rqv}{\kappa}
\newcommand{\mH}{\mathcal{H}}
\newcommand{\PAC}{\textsf{PAC}}

\newcommand{\appref}[1]{\iffull Appendix~\ref{#1}\else the full paper~\cite{kjoin}\fi}

\author{Simon Frisk}
\affiliation{
  \institution{University of Wisconsin--Madison}
  \city{Madison}
  \state{Wisconsin}
  \country{USA}
}
\email{simon.frisk@wisc.edu}

\author{Austen Fan}
\affiliation{
  \institution{University of Wisconsin--Madison}
  \city{Madison}
  \state{Wisconsin}
  \country{USA}
}
\email{afan@cs.wisc.edu}

\author{Paraschos Koutris}
\affiliation{
  \institution{University of Wisconsin--Madison}
  \city{Madison}
  \state{Wisconsin}
  \country{USA}
}
\email{paris@cs.wisc.edu}
  
\begin{document}


\title{$\kappa$-Join: Combining Vertex Covers for Parallel Joins}

\begin{abstract}
  Significant research effort has been devoted to improving the performance of join processing in the massively parallel computation model, where the goal is to evaluate a query with the minimum possible data transfer between machines. However, it is still an open question to determine the best possible parallel algorithm for any join query. In this paper, we present an algorithm that takes a step forward in this endeavor. Our new algorithm is simple and builds on two existing ideas: data partitioning and the HyperCube primitive. The novelty in our approach comes from a careful choice of the HyperCube shares, which is done as a linear combination of multiple vertex covers. The resulting load with input size $n$ and $p$ processors  is characterized as $n/p^{1/\kappa}$, where $\kappa$ is a new hypergraph theoretic measure we call the reduced quasi vertex-cover. The new measure matches or improves on the state-of-the-art and exhibits strong similarities to the edge quasi-packing that describes the worst-case optimal load in one-round algorithms.
\end{abstract}

\maketitle

\section{Introduction}

The MPC model, or \textit{massively parallel computation model}, is frequently used to theoretically study algorithms to compute database queries in clusters of processors \cite{communicationsteps}. In this model, an algorithm consists of several rounds, where each round consists of an exchange of data between processors, followed by local computation. Guided by the observation that performance is often bottlenecked by communication cost and synchronization cost, an algorithm in this model should minimize the number of rounds as well as the load of each round, defined as the maximum amount of data received by any processor in the round.

There has been extensive research into understanding the cost of computing a join query in the MPC model, as a function of the query, the input size $n$, and the number of processors $p$. Early work on the MPC model focused on algorithms with one round of communication \cite{communicationsteps} and input data without skew, and built on the HyperCube algorithm to process joins \cite{hypercube}. 
For arbitrary database instances, a tight bound of $n/p^{1/\psi^*}$, where $\psi^*$ is the edge quasi-packing, was obtained \cite{oneroundwcoj}. More recent work focused on algorithms where the number of communication steps is constant in $n$ and $p$. A tight bound of $n/p^{1/\rho^*}$, where $\rho^*$ is the fractional edge cover of the query, has been obtained for acyclic queries \cite{hutaoacyclic} and queries over binary relations \cite{ketsmansuciutaobinary}. Similar techniques have been used to improve the load on other classes of queries \cite{qiaotao2attribute}. The quantity $n/p^{1/\rho^*}$ is a known lower bound on all queries \cite{oneroundwcoj}, which builds on a counting argument using the AGM bound \cite{AGM}. A higher lower bound was described for the boat query, showing that $n/p^{1/\rho^*}$ cannot be matched in general \cite{hutaoacyclic}. A recent state-of-the-art improvement gave an algorithm matching the best known load on all queries except for the Loomis--Whitney join \cite{aamerketsmanpac}. Whereas prior work uses heavy-light partitioning, this more recent work uses a more fine-grained data partitioning.
Despite all these research efforts, the general worst-case optimal load remains open.

\paragraph{Our Contributions}

In this paper, we present a new join algorithm, called $\rqv$-Join, establishing an upper bound of $\tilde{O} (n/p^{1/\kappa})$ on the load required to compute a join query in the MPC model, where $\tilde{O}$ hides a polylogarithmic factor in $n,p$. Here, $\kappa$ is a graph-theoretic quantity defined over the hypergraph of the query, which we call the \textit{reduced quasi vertex-cover}. 
Our algorithm is relatively simple, but builds on several key ideas not found in previous work.

The definition of $\kappa$ bears strong similarities to the edge quasi-packing $\psi^*$, which describes the worst-case optimal load achievable with one-round algorithms in the MPC model. For a join query described by a hypergraph $\mH = (V,E)$, the edge quasi-packing is defined as $\psi^*:=\max_{S\subseteq V}\tau^*(\mH[S])$, where $\tau^*(\mH[S])$ denotes the value of the minimum fractional vertex cover on the subhypergraph induced by the vertex set $S$. Similarly, the new measure $\kappa$ is defined as $\max_{S\subseteq V}\tau^*(\textsf{red}(\mH[S]))$, meaning the only difference is that any relation that is contained in another relation in $\mH[S]$ is removed before finding the minimum vertex cover.

In addition, we present some arguments for why we believe that the derived upper bound of $\tilde{O} (n/p^{1/\kappa})$ might be optimal. Even though we do not show a tight lower bound for all queries, we provide a concrete worst-case instance construction that we conjecture to lead to it.  

\paragraph{Technical Overview}

Our algorithm builds on several concepts that exist in previous work, but also integrates multiple new key insights. A central primitive in the MPC model is the \textit{HyperCube} algorithm. Many algorithms in previous work have the following overarching approach \cite{aamerketsmanpac,oneroundwcoj,ketsmansuciutaobinary}. First, the input data is partitioned based on the degrees of different values, and the query is then computed in parallel for each partition. For each partition, certain join attributes, called \textit{heavy}, may have a high degree. Each possible tuple $t_H$ over the heavy variables gets assigned a dedicated set of machines, which compute the part of the output where the heavy variables have values $t_H$. The query is often computed using some semijoins and then invoking HyperCube.

Our algorithm also starts by performing data partitioning, but we perform a very fine-grained partitioning, which gives us precise control over the degrees occurring in each partition. To compute the join for a fixed data partition, our approach is quite different from previous work. We first identify and perform some binary joins that only grow the intermediate result size by a small factor. We then invoke the HyperCube algorithm on the intermediate results, instead of on the input relations. The shares used in the HyperCube algorithm are picked as a linear combination of vertex covers for different subqueries, each vertex cover corresponding to a hypergraph $\textsf{red}(\mH[S])$ for some subset $S \subseteq V$.

\paragraph{Comparison to State of the Art}
Our algorithm improves upon the state-of-the-art algorithm, known as \PAC \cite{aamerketsmanpac}, in two major ways. First, our algorithm achieves a load that is always at least as good as that of PAC. This is also true for other algorithms from previous work \cite{qiaotao2attribute,ketsmansuciutaobinary,hutaoacyclic}. Furthermore, there is at least one class of queries, the Loomis--Whitney joins, where our algorithm gives a strict improvement. We believe a fundamental reason for this is that previous algorithms allocate to each heavy tuple a dedicated set of machines, an approach we think is limited in general.
Second, the load of the \PAC algorithm is $n/p^{1/\gamma}$, where $\gamma$ is the \PAC-number. The \PAC-number has a complex definition, where given a query $Q$ it is non-trivial to determine what the \PAC-number of $Q$ is, and if a candidate \PAC-number is known, it is very hard to determine if a better one is possible. Furthermore, a lot of the complexity in the definition of $\gamma$ is inherited by the algorithm. In contrast, our algorithm has the load $n/p^{1/\kappa}$, where $\kappa$ has a straightforward hypergraph-theoretic definition and can be computed with a mixed integer linear program. Our algorithm is also simpler, removing many of the different cases in the \PAC{} algorithm and relying on fewer concepts.

\paragraph{Other Related Work}

Our algorithm uses partitioning, which is a central building block in worst-case optimal join algorithms in the RAM model \cite{nprr,skewstrikesback, triejoin}. The \textsf{PANDA} algorithm uses a fine-grained partitioning where degrees in each partition are uniformized \cite{panda}, similar to our algorithm.

An early paper used a partitioning method similar to ours \cite{joglekarrematterdegree}. This work studied distributed joins in a different model than the MPC model, and the load depended on $OUT$, which can be much more expensive in the worst case.
The Yannakakis algorithm \cite{Yannakakis81} can be applied in the context of the MPC model to compute acyclic queries with load $O((IN+OUT)/p)$ \cite{afratigym}. This has been improved to $O((IN+\sqrt{IN\cdot OUT})/p)$ \cite{huyioutputoptimal}, in work that studied upper and lower bounds on instance- and output-optimal evaluation.

A related line of work generalized the MPC model to optimize query processing with respect to the topology of the compute cluster \cite{hu2020algorithms, topologyawaredataprocessing, topologyawarejoins,heterogeneousmachines}. Finally, there has also been progress on implementing ideas from theoretical MPC research into concrete systems \cite{HCsigmod15, zhangyuparallel, automj}.

\paragraph{Paper Organization}
The paper is organized as follows. 
In Section~\ref{sec:prelim}, we give preliminary definitions. 
In Section~\ref{sec:measure}, we define the new measure $\kappa$, the \textit{reduced quasi vertex-cover}. 
In Section~\ref{sec:upperbound}, we give an upper bound showing how to compute a join query with load $\tilde{O}(n/p^{1/\kappa})$.
In Section~\ref{sec:lowerbounds}, we discuss lower bounds and tightness of our result, establishing a path for future research to find what the best worst-case optimal load is.
\section{Preliminaries}
\label{sec:prelim}

A {\em relation} $R$ with schema $\vars{R}$ is a set of tuples, where 
each tuple maps each variable in $\vars{R}$ to an element of the domain $\mathbf{dom}$. We will use $x,y,z, \dots$ to denote variables. A {\em natural join query} $Q$ can be defined as a set of relations, denoted as $\Join_i R_i$. We will use $\vars{Q}$ to denote the variables that occur across all relations in $Q$.  The input size of a query is defined as the total number of tuples across all relations.

We will assume that $Q$ has no self-joins, i.e., no relations can repeat in the query. This assumption is w.l.o.g. for any upper bound, since we can simply create a copy of the repeated relation. We will also assume w.l.o.g. that no two relations can have the same schema.


\paragraph{Hypergraphs and Measures}

We can now associate with each natural join query $Q$ a hypergraph $\mathcal{H} = (V,E)$ such that $V = \vars{Q}$ and $E = \{\vars{R} \mid R \in Q\}$.
For $S \subseteq V$, we define $\mH[S]$ to be the vertex-induced subhypergraph of $\mH$, i.e., the hypergraph $\mH'$ with vertex set $S$ and edge set $\{S \cap e \mid e \in E, e \cap S \neq \emptyset \}$.

\begin{definition}[Reduced Hypergraph]
    Let $\mH$ be a hypergraph. The reduced hypergraph, denoted $\textsf{red}(\mH)$, is a hypergraph $(V,E')$, where $E'\subseteq E$ is a subset obtained by removing any edge $e\in E$ such that there is another $e'\in E$ with $e \subseteq e'$.
\end{definition}

In combinatorics, a \textit{clutter} or \textit{Sperner family} is a pair $(V,E)$, where $V$ is a set and $E$ is a set of subsets of $V$, such that for no distinct $e,e'\in E$ is $e\subseteq e'$. A reduced query forms a Sperner family.

\begin{definition}[Edge Packing \& Covering]
Let $\mH = (V,E)$ be a hypergraph. Let $\ve{u}$ be a (non-negative) weight assignment on the hyperedges $E$. Then, $\ve{u}$ is a:
\begin{itemize}
\item {\em fractional edge packing (or matching)} if for every vertex $v \in V$, $\sum_{v \in e} u_e \leq 1$. The total weight of the maximum fractional edge packing is denoted by $\tau^*(\mH)$.
\item {\em fractional edge covering} if for every vertex $v \in V$, $\sum_{v \in e} u_e \geq 1$. The total weight of the minimum fractional edge covering is denoted by $\rho^*(\mH)$.
\end{itemize}
\end{definition}

In general, there is no relationship between the two quantities. We will also be interested in {\em fractional vertex covers} of a hypergraph $\mH$: these are weight assignments $\ve{v}$ on the vertices $V$ such that for every hyperedge $e \in E$, $\sum_{v \in e} \ve{v}_v \geq 1$. By the duality principle in linear programming, the total weight of the minimum fractional vertex cover is equal to $\tau^*(\mH)$. When we mention covers and packings, we are always referring to fractional and not integral ones.

\paragraph{The MPC Model}

The MPC model, or \textit{Massively Parallel Computation model}, is a theoretical model to analyze the cost of distributed algorithms. It uses a cluster of $p$ machines, and the input, whose size is $n$, is initially partitioned across the machines, with each machine storing $n/p$ data in local memory. The computation is synchronous and proceeds in $r$ rounds, where each round consists of message passing between the machines, and computation on data in local memory. By the end of the algorithm, we require that the union of the outputs of all machines forms the desired output. The cost of an algorithm is modeled as the number of rounds $r$ and the maximum data size $L$ that any machine receives in any round. An ideal MPC algorithm works in a constant number of rounds and achieves linear load, i.e., $L = O(n/p)$. 

We are interested in using the MPC model to analyze join processing, and in particular we focus on the regime where the number of rounds $r$ is constant in $n$ and $p$, but may depend on the query. We will also work in the regime where $n \gg p$, i.e., the data size is much bigger than the number of processors (we will specify the exact requirement in each theorem statement).
Given a hypergraph $\mH$, our main question is: {\em what is the smallest exponent $\epsilon(\mH)$ such that for every natural join query $Q$ with hypergraph $\mH$ and input size at most $n$, we can compute its output in constant rounds and load $O(n/p^{1/\epsilon(\mH)})$?} This is what we call a worst-case optimal setting, since the goal is to match the cost of the worst-behaved instance with size $\leq n$.

\paragraph{The HyperCube Algorithm}
The HyperCube algorithm is a well-known algorithm for partitioning multiple relations across machines,
such that the join of the relations can be computed locally. Given a query $Q$ with hypergraph $\mH=(V,E)$,
the HyperCube algorithm uses \textit{share exponents} $\ve{v} = (v_x)_{x \in V}$, where $v_x \geq 0$ for all
$x$ and $\sum_{x \in V} v_x \leq 1$. The algorithm organizes the $p$ machines into a $|V|$-dimensional hypercube,
where the dimension of variable $x \in V$ is $p^{v_x}$ machines. A machine $m$ then has a coordinate vector $(m_x)_{x \in V}$
in the hypercube. Using hash functions $\{h_x\}_{x \in V}$,
the algorithm partitions tuples, by sending a tuple $t$ over variables $X \subseteq V$ to all machines $m$
where $h_x(t[x]) = m_x$ for all $x \in X$.

\section{A New Hypergraph Measure for MPC}
\label{sec:measure}

In this section, we define our new measure of the hypergraph of a join query. 

\begin{definition}\label{def:kappa}
    Let $\mH=(V,E)$ be a hypergraph. The {\em reduced quasi vertex-cover} is defined as
    \[
    \rqv(\mH) :=\max_{S \subseteq V}\tau^*(\textsf{red}(\mH[S])).
    \]
\end{definition}

In other words, to compute $\rqv$ we need to find the largest $\tau^*$ of any vertex-induced subhypergraph, after we reduce it. The last part is critical to the definition of $\kappa$, since otherwise we would obtain $\psi^*(\mH) = \max_{S \subseteq V}\tau^*(\mH[S])$, the edge quasi-packing of $\mH$. In Appendix~\ref{sec:kappamilp}, we show how $\kappa(\mH)$ can be calculated by expressing it as the objective of a mixed integer linear program.

\begin{example}
    Consider the hypergraph $\mH_1$ with hyperedges $\{\{x,y\},\{y,z\}\}$.
    Then, $\tau^*(\mH_1) = 1$, but $\tau^*(\mH_1[\{x,z\}])=2$. Hence, $\rqv(\mH_1)=2$.
\end{example}

\begin{example}
    Consider the hypergraph $K_k$ corresponding to the $k$-clique query. In this case, $\rqv(K_k)=k/2$. It is useful to compare this to $\psi^*(K_k)= k-1$.
\end{example}

\begingroup
\refstepcounter{example}\label{ex:boatzquery}
\par\addvspace{\topsep}\noindent
\strut\textsc{Example~\theexample.}\itshape\ Consider the hypergraph $\mH_3^\dagger$ with the following hyperedges:
\[ \{ \textcolor{eBlue}{\{ x_1,x_2,x_3,z\}}, \textcolor{eRed}{\{y_1,y_2,y_3,z\}}, \textcolor{eYellow!60!black}{\{x_1, y_1, z \}}, \textcolor{eAqua!70!black}{\{x_2, y_2, z\}}, \textcolor{eViolet}{\{x_3, y_3,z\}} \}. \]
\par\centering
\resizebox{0.3\linewidth}{!}{\scalebox{0.7}{
\begin{tikzpicture}[
  vertex/.style={circle,draw,inner sep=2pt,minimum size=6mm},
  hedge/.style={
    fill opacity=0.3,
    draw opacity=1,
    line width=1pt,
    rounded corners=4pt
  }
]

\node[vertex] (x1) at (0,2) {$x_1$};
\node[vertex] (x2) at (0,1) {$x_2$};
\node[vertex] (x3) at (0,0) {$x_3$};
\node[vertex] (y1) at (4,2) {$y_1$};
\node[vertex] (y2) at (4,1) {$y_2$};
\node[vertex] (y3) at (4,0) {$y_3$};
\node[vertex] (z)  at (2,1) {$z$};

\path[hedge,fill=eBlue,draw=eBlue,rounded corners=10pt]
  ($(x1)+(-0.4,0.7)$) --
  ($(z)+(1,0)$) --
  ($(x3)+(-0.4,-0.7)$) -- cycle;

\path[hedge,fill=eRed, draw=eRed,rounded corners=10pt]
  ($(y1)+(0.4,0.7)$) --
  ($(z)+(-1,0)$) --
  ($(y3)+(0.4,-0.7)$) -- cycle;

\path[hedge,fill=eAqua,draw=eAqua,rounded corners=10pt]
  ($(x2)+(-0.4,0.4)$) --
  ($(y2)+(0.4,0.4)$) --
  ($(y2)+(0.4,-0.4)$) --
  ($(x2)+(-0.4,-0.4)$) -- cycle;

\path[hedge,fill=eYellow,draw=eYellow,rounded corners=10pt]
  ($(x1)+(-0.6,-0.1)$) --
  ($(x1)+(-0.4,0.5)$) --
  ($(y1)+(0.4,0.5)$) --
  ($(y1)+(0.6,-0.1)$) --
  ($(z)+(0,-0.7)$) --
  cycle;

\path[hedge,fill=eViolet,draw=eViolet,rounded corners=10pt]
  ($(x3)+(-0.6,0.1)$) --
  ($(x3)+(-0.4,-0.5)$) --
  ($(y3)+(0.4,-0.5)$) --
  ($(y3)+(0.6,0.1)$) --
  ($(z)+(0,0.7)$) --
  cycle;

\node[vertex,fill=white] at (0,2) {$x_1$};
\node[vertex,fill=white] at (0,1) {$x_2$};
\node[vertex,fill=white] at (0,0) {$x_3$};
\node[vertex,fill=white] at (4,2) {$y_1$};
\node[vertex,fill=white] at (4,1) {$y_2$};
\node[vertex,fill=white] at (4,0) {$y_3$};
\node[vertex,fill=white] at (2,1) {$z$};

\end{tikzpicture}
}}
\par\noindent
Here, $\rqv(\mH_3^\dagger)=3$, which is achieved by the subset of variables $S=\{x_1,x_2,x_3,y_1,y_2,y_3\}$, since $\mH_3^\dagger[S]$ is precisely the boat query \cite{hutaoacyclic} with $\tau^*=3$.\strut
\par\addvspace{\topsep}
\endgroup

Whenever it is clear from the context, we will use $\rqv$ instead of $\rqv(\mH)$. Below, we use the term connected component in the standard graph-theoretic sense.

\begin{lemma}
\label{lemma:KappaCCSum}
    Let $\mH=(V,E)$ be a hypergraph with connected components $\mH_1=(V_1,E_1),\dots,\mH_t=(V_t,E_t)$. Then, $\rqv(\mH)=\sum_{i=1}^t\rqv(\mH_i)$.
\end{lemma}
\begin{proof}
    Let $S_1\subseteq V_1, \dots, S_t\subseteq V_t$ be the sets that maximize $\tau^*(\textsf{red}(\mH_1[S_1])), \dots,  \tau^*(\textsf{red}(\mH_t[S_t]))$ respectively. We claim that the $S\subseteq V$ that maximizes $\tau^*(\textsf{red}(\mH[S]))$ is $S=S_1\cup\dots\cup S_t$. Indeed, if there exists another $S'\neq S$ where $\tau^*(\textsf{red}(\mH[S'])) > \tau^*(\textsf{red}(\mH[S]))$, then there has to exist some $i\in\{1,\dots,t\}$ such that $S_i$ does not maximize $\tau^*(\textsf{red}(\mH_i[S_i]))$.
\end{proof}

We next connect the newly defined measure to existing hypergraph quantities. 

\begin{restatable}{lemma}{kappacomparison}
\label{lemma:kappacomparison}
    For any hypergraph $\mH = (V,E)$:
    \begin{enumerate}
        \item $\rqv(\mH) \geq \tau^*(\textsf{red}(\mH));$
        \item $\rqv(\mH) \geq \rho^*(\textsf{red}(\mH)) = \rho^*(\mH);$
        \item $\rqv(\mH) \leq \psi^*(\mH).$
    \end{enumerate}
\end{restatable}



Next, we give some lemmas that allow us to compare $\kappa$ to quantities used to characterize the load of algorithms from previous work. This helps us establish that the load $n/p^{1/\kappa(\mH)}$ is at least as good as that of previous work.

\begin{restatable}{lemma}{acyclickapparho}
\label{lem:acyclic}
For any acyclic hypergraph $\mH$, $\rqv(\mH) = \rho^*(\mH)$. 
\end{restatable}

\begin{restatable}{lemma}{kapparhorange}
\label{lem:kapparhorange}
For any hypergraph $\mH$, $\rho^*(\mH) \leq \rqv(\mH)$. Moreover, if the maximum arity $\alpha$ among hyperedges satisfies $\alpha \geq 2$, then $\rho^*(\mH) \leq \rqv(\mH) \leq\frac{\alpha}{2}\rho^*(\mH)$.
\end{restatable}
The above implies that a hypergraph $\mH$ with binary relations has $\kappa(\mH)=\rho^*(\mH)$. Additionally, for all hypergraphs with $\alpha \geq 2$, $\kappa \leq \frac{\alpha\phi}{2}$, where $\phi$ is the \textit{generalized fractional vertex packing} from previous work \cite{qiaotao2attribute}, since $\phi \geq \rho^*$.
We can also compare $\rqv$ to the quantity $\PAC$ \cite{aamerketsmanpac, pacv2}. The proof of this lemma is given in Appendix~\ref{sec:comparepac:app}.
\begin{restatable}{lemma}{comparisonkappapac}
\label{lemma:comparisonkappapac}
For any hypergraph $\mH$, $\rqv(\mH) \leq \PAC(\mH).$
\end{restatable}
%
In \appref{sec:dedicatedmachines:app}, we discuss an example indicating why allocating dedicated machines to heavy tuples, as done in $\PAC$ and other previous work, has limitations. In Appendix~\ref{sec:compareclasspacv2}, we compare $\kappa$ to specific upper bounds given in a follow-up paper by the PAC authors \cite{pacv2}.

Finally, we consider whether $\rqv$ is equal to the maximum of $\rho^*(\mH)$ and $\tau^*(\mH)$. We answer this in the negative with the following family of counterexamples. For any integer $k>0$, consider the hypergraph $\mH_k^\dagger$ with vertex set $\{x_1, \dots, x_k, y_1, \dots, y_k, z\}$ and the following set of hyperedges: 
\[ \{ \{ x_1,x_2,\dots, x_k,z\}, \{y_1,y_2,\dots, y_k,z\}, \{x_1, y_1, z \}, \{x_2, y_2, z\}, \dots, \{x_k, y_k,z\} \} \]
$\mH_k^\dagger$ is the generalized boat join~\cite{hutaoacyclic} where we have added a variable $z$ to every hyperedge. $\mH_3^\dagger$ is visualized in Example~\ref{ex:boatzquery}.

\begin{restatable}{proposition}{boatzqueryquantities}
\label{boatzqueryquantities}
For the hypergraph $\mH_k^\dagger$, $k \geq 2$, we have:
    \begin{enumerate}
        \item $\rho^*(\mH_k^\dagger)=2;$
        \item $\tau^*(\mH_k^\dagger)=1;$
        \item $\rqv(\mH_k^\dagger)=k.$
    \end{enumerate}
\end{restatable}

Hence, there exists a family of hypergraphs where $\rqv$ grows asymptotically faster than $\max \{\rho^*, \tau^* \}$.

\section{Upper Bound}
\label{sec:upperbound}

In this section, we present a new MPC algorithm that achieves load $\tilde{O}(n/p^{1/\rqv})$.

\subsection{A Partitioning Step}

Let $R$ be a relational instance with schema $\vars{R}$. Let $X,Y \subseteq \vars{R}$ be two subsets of variables. Let $\ve{y}$ be a tuple over $Y$. Define
$$\deg_R(X \mid Y=\ve{y}) := |\pi_X (\sigma_{Y=\ve{y}}R)| .$$
Furthermore, define $\deg_R(X \mid Y) := \max_{\ve{y}}\deg_R(X|Y=\ve{y})$. For convenience, we will sometimes use the shorthand notation $\deg_R(X)=\deg_R(\vars{R} \mid X)$.


Next, we define a type of constraint. This constraint is similar to the property we get in a uniformized instance, where we would replace the inequalities below with equalities. Given a set $X$, we denote by $\mathcal{P}(X)$ the set of all subsets of $X$, including $\emptyset$ and $X$ itself.

\begin{definition}
A {\em constraint set} for a relation $R$ is a total function $\theta: \mathcal{P}(\vars{R}) \rightarrow \mathbb{R}^{\geq 0}$.
\end{definition}

\begin{definition}
\label{def:sigmaconstraint}
We say that a relational instance $R$ satisfies a constraint set $\theta_R$ if for every pair of subsets $X \subseteq Y \subseteq \vars{R}$ it holds:
    $$ \deg_R(Y \mid X) \leq 2 \frac{\theta_R(X)}{\theta_R(Y)}.$$
\end{definition}

An intuitive mental model of $\theta_R(X)$, which is not mathematically precise but nonetheless useful, is as a measure of how many tuples there are in $R$ for each tuple over $X \subseteq \vars{R}$.
This is why we will always have $\theta_R(\emptyset)=n$ and $\theta_R(\vars{R}) = 1$.

Below we see a partitioning method for a relation $R$. Let $\mathcal{S} \subseteq\mathcal{P}(\vars{R})$ be an ordered set of subsets that is sorted in increasing order of cardinality.

\begin{algorithm}[H]
\caption{Partitioning Procedure}
\label{algorithm:UniformizePartition}
    \begin{algorithmic}[1]
    \Function{Partition}{$R: \text{relational instance}, \mathcal{S}: \text{remaining sets } \subseteq \mathcal{P}(\vars{R})$}
    \If{$\mathcal{S}=\emptyset$}
        \Return $\{ R \}$
    \EndIf
    \State $U \gets \text{first item in } \mathcal{S}$
    \State $\forall i\in[\log n]: R^i \gets \sigma_{U \in \{u|\deg_R(U=u)\in[2^{i-1},2^i]\}}R$
    \State\Return $\bigcup_{i\in[\log n]}\textsc{Partition}(R^i, \mathcal{S}-\{U\})$
    \EndFunction
    \end{algorithmic}
\end{algorithm}

In the above, some partitions may be empty, if the constraints are picked in an incompatible way. Such partitions can be disregarded.

\begin{lemma}\label{lemma:UniformizationProperties}
Let $\mathcal{R}$ be the output of Algorithm~\ref{algorithm:UniformizePartition} when it runs on a relation $R$ of size $n$ and $\mathcal{S}=\mathcal{P}(\vars{R})$. Then,
\begin{enumerate}
    \item $|\mathcal{R}| \leq \log^{|\mathcal{P}(\vars{R})|}n$;
    \item every instance $T \in \mathcal{R}$ satisfies a constraint set $\theta_T$ with $\theta_T(\emptyset) = n$ and $\theta_T(\vars{R}) = 1$. 
    \end{enumerate}
\end{lemma}

\begin{proof}
We show both items of this lemma in sequence.

\textbf{1)} In each step, we create at most $\log n$ subinstances, and the algorithm runs recursively for at most $|\mathcal{P}(\vars{R})|$ steps. 

\textbf{2)}
We will construct the constraint set for any subinstance following the recursive process. Consider the step where the algorithm partitions on some set $U \subseteq \vars{R}$: then, for the subinstance $R^i$, we will pick $\theta(U) = 2^i$. Observe that during the first step, when $U = \emptyset$, all tuples fall in the subinstance with the largest value, hence $\theta(\emptyset) = n$. During the last step, when $U=\vars{R}$, all degrees are $1$ and hence every subinstance will set $\theta(\vars{R})=1$.

It remains to show that the desired inequality holds for any subinstance $T \in \mathcal{R}$ and every $X \subseteq Y \subseteq \vars{R}$. Indeed, consider the step of the algorithm just after $R$ was partitioned on $Y$. Let $\ve{x}$ be any tuple over $X$. For any tuple $\ve{y}$ over $Y$ where $\pi_{X}(\ve{y})=\ve{x}$, since we just partitioned using $Y$, $\deg_R(\vars{R} \mid Y=\ve{y}) \geq \theta(Y)/2$. Thus, we can write:
    $$ \deg_R(\vars{R} \mid X=\ve{x})=\sum_{\ve{y}\in Y|\pi_{X}(\ve{y})=\ve{x}}\deg_R(\vars{R} \mid Y=\ve{y})
    \geq \frac{\theta_{R}(Y)}{2} \cdot \deg_{R}(Y \mid X=\ve{x}).  $$
We now claim that $\deg_R(\vars{R} \mid X=\ve{x}) \leq \theta_R(X)$. This holds since the inequality $\deg_R(\vars{R} \mid X=\ve{x})\leq \theta_R(X)$ holds right after we partitioned on $X$, and when we repartition the degree can only grow smaller. Reorganizing we obtain:
$$ \deg_{R}(Y \mid X=\ve{x}) \leq 2\frac{\theta_R(X)}{\theta_R(Y)}.$$
Finally, observe that this inequality will hold after Algorithm~\ref{algorithm:UniformizePartition} has completed, since the degree can only become smaller by removing tuples during repartitioning.
%
\end{proof}

Suppose we are given a query $Q$. Let $\Theta = \{\theta_R\}_{R \in Q}$ be a set of constraint sets, one for each relation $R$ in the query. We will say that $Q$ is $\Theta$-uniformized if every relation $R$ in $Q$ satisfies the constraint set $\theta_R \in \Theta$. Lemma~\ref{lemma:UniformizationProperties} tells us that for any input instance, we can create polylogarithmically many subinstances such that each one is $\Theta$-uniformized for some chosen $\Theta$. Thus, from here on we will focus on a $\Theta$-uniformized query $Q$. By dividing the instance into $\Theta$-uniformized subinstances, we will be able to tailor the algorithm on each subinstance to the specific degrees occurring in the subinstance, which is a key idea behind our algorithm.

\paragraph{Partitioning in the MPC model}

We finally describe how to perform Algorithm~\ref{algorithm:UniformizePartition} in the MPC model. We will do this with a number of rounds which is constant in $n$ and $p$, and load $O(n/p)$. We will use the well-known primitive \textsf{sum-by-key} to do this, which is described in previous work \cite{mpcsimilarityjoins}.

Let $R$ be a relation with $n$ tuples stored on $p$ machines, where each tuple $t$ is associated with a key $\textsf{key}(t)$. \textsf{sum-by-key} computes, for each key $\Bar{k}$, the total number $\textsf{count}(\Bar{k})$ of tuples $t$ where $\textsf{key}(t)=\Bar{k}$. Afterwards, if a machine stores tuple $t$, it also knows $\textsf{count}(\textsf{key}(t))$. 

\begin{proposition}
    \textsf{sum-by-key} can be implemented in $O(1)$ rounds with load $n/p$, assuming $n>p^{1+\epsilon}$ for some small constant $\epsilon>0$.
\end{proposition}

We can now perform each step in Algorithm~\ref{algorithm:UniformizePartition} by performing \textsf{sum-by-key}, where for a subset $U$ of the variables, the key of a tuple $t$ is $\pi_Ut$. When the machines have $\textsf{count}(\textsf{key}(t))$ for each tuple $t$ they store, they can locally decide which partition $t$ fits in.

\subsection{Vertex Weight Mappings}

We define a {\em vertex weight mapping} of a query $Q$ with hypergraph $\mH = (V,E)$ to be a function $v: V \rightarrow\mathbb{R}_{\geq 0}$. We will roughly use these as vertex covers are used in previous work, but we may leave some relations not covered.

\begin{definition}
\label{def:lightheavyconsistent}
Let $\ve{v}$ be a vertex weight mapping of a $\Theta$-uniformized query $Q$ with hypergraph $\mH$. For a relation $R$ in $Q$ of size $n$ and $U \subseteq \vars{R}$, we say that $U$ is:
\begin{itemize}
    \item light in $R$ if $\theta_R(U)<\frac{n}{p^{\sum_{i\in U}v_i}}$;
    \item heavy in $R$ if $\theta_R(U) \geq \frac{n}{p^{\sum_{i\in U}v_i}}$;
    \item consistent in $R$ if $\theta_R(U) \leq \frac{n}{p^{\sum_{i\in U}v_i}}$.
\end{itemize}

\end{definition}

Here, we will get $\theta$ as described in Lemma~\ref{lemma:UniformizationProperties}.  Notice that the difference between light and consistent is merely the strictness of the inequality. We can also view consistent as being light or heavy with equality. This will be useful later on.

The above definition and Lemma~\ref{lemma:UniformizationProperties} tell us that for a light $U$ in $R$, the degree of $U$ in $R$ is bounded as $\deg_R(U) \leq 2 \theta_R(U) < 2{n}/{p^{\sum_{i\in U}v_i}}$. In addition, for any relation $R$, the empty set is always heavy but consistent, and $\vars{R}$ is always light, by the assumption $n>p^2$.

For readers familiar with previous work, we note that we will often normalize our vertex weight mappings for convenience. This means we will often have terms of the form $n/p^{v_x}$, instead of $n/p^{v_x/\sum_iv_i}$ as in some previous work.



\begin{definition}[Consistent Vertex Weight Mapping]
 Let $Q$ be a $\Theta$-uniformized query, and $\ve{v}$ be a vertex weight mapping of $Q$. We say that $\ve{v}$ is {\em consistent} with $Q$ if for every $R\in Q$ and every $U\subseteq \vars{R}$, $U$ is consistent in $R$.
\end{definition}

Our goal is to use HyperCube with shares picked using a consistent weight mapping, because then there are no values that can individually break the load requirement. Consistency means that the allocated shares are not too big---as shares increase, the heavy thresholds decrease. This means there is some point when the shares can become too big. Later, we will also need that the shares are not too small, so the load is not too high.

\begin{definition}[Heavy Sets]
 Let $Q$ be a $\Theta$-uniformized query, and $\ve{v}$ be a vertex weight mapping of $Q$. We denote by $\mathcal{U}[{\ve{v}}] = \{(U,R) \mid U\subseteq \vars{R} \text{ and } U \text{ is heavy in } R\}$ the collection of all heavy sets in $Q$, and by $H[\ve{v}] = \bigcup_{(U,R) \in \mathcal{U}[{\ve{v}}]} U$ the set of all variables that occur in some heavy set.
\end{definition}

\paragraph{Constructing Good Vertex Weight Mappings}
In this part, we will show a simple procedure that picks a suitable vertex weight mapping for our algorithm.
We introduce the following notation: for any $U\subseteq V$, let $\ve{v}^*_{U}$ be a minimum fractional vertex cover of $\textsf{red}(\mH[U])$.

\begin{definition}
    Consider the vector space $\VCSpace=\mathbb{R}^{|\mathcal{P}(V)|}$. Let $\VCFunc: \VCSpace\rightarrow\mathbb{R}^{|V|}$ be the function that maps $(\VCVec_U)_{U\in\mathcal{P}(V)} \in\VCSpace$ to $\sum_{U\in\mathcal{P}(V)}\VCVec_U\ve{v}^*_U$. 
\end{definition}


For some $\VCVec=(\VCVec_U)_{U\in\mathcal{P}(V)}\in \VCSpace$, we denote by $|\VCVec|_{1}=\sum_{U\in\mathcal{P}(V)}\VCVec_U$, in the standard way. We denote by $\VCFunc_x$ the component function of $\VCFunc$ on element $x$.
We are interested in finding vectors $\VCVec\in \VCSpace$ where $|\VCVec|_{1} \leq 1$ and $\VCVec_U\geq 0$ for all $U$. Each such vector can be thought of as describing a way to pick shares for a HyperCube partitioning with at most $p$ machines, by assigning to variable $x$ a share $p^{\VCFunc_x(\VCVec)/\Delta}$, for some parameter $\Delta \geq |\VCFunc(\VCVec)|_1$ (the reader can think of $\Delta$ as $\rqv$ for now). In other words, we pick share exponents as linear combinations of minimum fractional vertex covers, and $\VCVec$ is the coefficient vector describing the linear combination.
The following algorithm describes how we pick the linear combination $\VCVec\in\VCSpace$. 

\begin{algorithm}[H]
    \caption{Procedure to Construct $\VCVec$}
    \label{algorithm:FindVertexWeights}
    \begin{algorithmic}[1]
    \Function{Construct}{$Q: \Theta\text{-uniformized instance over }\mH, \Delta\in\mathbb{R}_{> 0}$}
    \State $\VCVec \gets \ve{0}$
    \While{$|\VCVec|_1 < 1$ and $H[\VCFunc(\VCVec)/\Delta] \neq V$}
        \State ${L} \gets V\setminus {H}[\VCFunc(\VCVec)/\Delta]$
        \State $\VCVec_{L} \gets \max_{\delta\in(0, 1-|\VCVec|_1], \text{ such that $Q$ is consistent with $(\VCFunc(\VCVec)+\delta \cdot\ve{v}^*_{L})/\Delta$}}\delta$
    \EndWhile
    \State \Return $\psi$
    \EndFunction
    \end{algorithmic}
\end{algorithm}

\begingroup
\refstepcounter{example}\label{example:runningexpickweights}
\par\addvspace{\topsep}\noindent
\strut\textsc{Example~\theexample.}\itshape\ Consider the triangle query, whose hypergraph $\mH_\Delta$ has hyperedges ${\color{eRed}R=\{X,Y\}}$, ${\color{eGreen}S=\{Y,Z\}}$, ${\color{eBlue}T=\{Z,X\}}$. Consider a fixed $\Theta$-uniformized instance where $X$ has $p^{1/6}$ different values, and $Y$ and $Z$ each have $n/p^{1/6}$ values. $R$ and $T$ are product instances, and $S$ has $n$ tuples, where the degrees at both $Y$ and $Z$ are at most $n/p$.

\par\addvspace{0.4\baselineskip}
{\centering\begin{tikzpicture}[
    dot/.style={circle,fill=black,inner sep=0pt,minimum size=2.4pt},
    vbox/.style={fill=black!7,draw=black!35,line width=0.5pt,rounded corners=2pt},
    vlab/.style={font=\small},
    alab/.style={font=\scriptsize},
    rel/.style={line width=0.9pt}
]
\draw[vbox] (-0.21,0.83) rectangle (0.21,1.57);
\draw[vbox] (-2.01,-0.74) rectangle (-1.59,0.24);
\draw[vbox] (1.59,-0.74) rectangle (2.01,0.24);

\coordinate (a1) at (0.00,1.40);
\coordinate (a2) at (0.00,1.00);
\coordinate (b1) at (-1.80,0.07);
\coordinate (b2) at (-1.80,-0.17);
\coordinate (b3) at (-1.80,-0.57);
\coordinate (c1) at (1.80,0.07);
\coordinate (c2) at (1.80,-0.17);
\coordinate (c3) at (1.80,-0.57);

\draw[rel,eRed] (a1) -- (b1);
\draw[rel,eRed] (a1) -- (b2);
\draw[rel,eRed] (a1) -- (b3);
\draw[rel,eRed] (a2) -- (b1);
\draw[rel,eRed] (a2) -- (b2);
\draw[rel,eRed] (a2) -- (b3);
\draw[rel,eBlue] (a1) -- (c1);
\draw[rel,eBlue] (a1) -- (c2);
\draw[rel,eBlue] (a1) -- (c3);
\draw[rel,eBlue] (a2) -- (c1);
\draw[rel,eBlue] (a2) -- (c2);
\draw[rel,eBlue] (a2) -- (c3);
\draw[rel,eGreen] (b1) -- (c1);
\draw[rel,eGreen] (b1) -- (c2);
\draw[rel,eGreen] (b2) -- (c2);
\draw[rel,eGreen] (b3) -- (c3);

\node[dot] at (0.00,1.40) {};
\node[dot] at (0.00,1.00) {};
\node[dot] at (-1.80,0.07) {};
\node[dot] at (-1.80,-0.17) {};
\node[dot] at (-1.80,-0.57) {};
\node[dot] at (1.80,0.07) {};
\node[dot] at (1.80,-0.17) {};
\node[dot] at (1.80,-0.57) {};
\node[vlab] at (0.00,1.20) {$\vdots$};
\node[vlab] at (-1.80,-0.36) {$\vdots$};
\node[vlab] at (1.80,-0.36) {$\vdots$};

\node[vlab] at (0.00,2.16) {$X$};
\node[alab] at (0.00,1.79) {$p^{1/6}$ values};
\node[vlab] at (-1.80,-0.98) {$Y$};
\node[alab] at (-1.80,-1.35) {$n/p^{1/6}$ values};
\node[vlab] at (1.80,-0.98) {$Z$};
\node[alab] at (1.80,-1.35) {$n/p^{1/6}$ values};
\node[vlab,text=eRed] at (-1.07,0.81) {$R$};
\node[vlab,text=eBlue] at (1.07,0.81) {$T$};
\node[vlab,text=eGreen] at (0.00,-0.36) {$S$};
\end{tikzpicture}
\par}
\addvspace{0.4\baselineskip}

\hspace{1em}Let $\Delta=3/2$. Algorithm~\ref{algorithm:FindVertexWeights} will do the following: In the first iteration, we use the vertex cover $\ve{v}^*_{XYZ}=(0.5,0.5,0.5)$, and $\psi_{XYZ}=1/2$. This gives $\deg_R(\textsf{vars}(R)|X)=n/p^{1/6}=n/p^{(\ve{v}^*_{XYZ})_X\psi_{XYZ}/\Delta}$, meaning $\{X\}$ is heavy but consistent in both $R$ and $T$. In the second iteration, we use the vertex cover $\ve{v}^*_{YZ}=(0,0.5,0.5)$, with $\psi_{YZ}=1/2$. In total, we get the vertex weight mapping given by the following linear combination:
\[ \VCFunc(\psi)
={\color{gray}\underbrace{\color{black}\tfrac{1}{2}}_{\mathclap{\scriptstyle \psi_{XYZ}}}}{\color{gray}\underbrace{\color{black}(0.5,0.5,0.5)}_{\mathclap{\scriptstyle \ve{v}^*_{XYZ}}}}+{\color{gray}\underbrace{\color{black}\tfrac{1}{2}}_{\mathclap{\scriptstyle \psi_{YZ}}}}{\color{gray}\underbrace{\color{black}(0,0.5,0.5)}_{\mathclap{\scriptstyle \ve{v}^*_{YZ}}}}
=({\color{gray}\underbrace{\color{black}0.25}_{\mathclap{\scriptstyle \VCFunc_X(\psi)}}},{\color{gray}\underbrace{\color{black}0.5}_{\mathclap{\scriptstyle \VCFunc_Y(\psi)}}},{\color{gray}\underbrace{\color{black}0.5}_{\mathclap{\scriptstyle \VCFunc_Z(\psi)}}}). \]\strut
\par\addvspace{\topsep}
\endgroup

We first show that Algorithm~\ref{algorithm:FindVertexWeights} will always terminate and return some vector $\psi$.

\begin{lemma}
Algorithm~\ref{algorithm:FindVertexWeights} always terminates. 
\end{lemma}

\begin{proof}
We claim that as long as $|\VCVec|_1 < 1$, at each iteration the set of heavy variables $H[\VCFunc(\VCVec)/\Delta]$ increases by at least one variable. Indeed, if this was not true, it would be possible to pick a larger value of $\delta$. Therefore, the maximum number of iterations is $|V|$.
\end{proof}

\begin{lemma}\label{lem:vector:properties}
    Let $Q$ be a $\Theta$-uniformized query with hypergraph $\mH=(V,E)$. Let $\VCVec$ be the vector returned from Algorithm~\ref{algorithm:FindVertexWeights}. Then,
    \begin{enumerate}
        \item $|\VCVec|_1 \leq 1$;
        \item the vertex weight mapping $\VCFunc(\VCVec)/\Delta$ is consistent with $Q$;
        \item $|\VCFunc(\VCVec)|_1\leq \rqv(\mH)$.
    \end{enumerate}
\end{lemma}
\begin{proof} We separately prove the items in the lemma.

    \textbf{(1)} This follows directly from the termination condition of the while loop.
    
    \textbf{(2)} We inductively argue that before and after each iteration, $\VCFunc(\VCVec)/\Delta$ is consistent w.r.t. $Q$.
    
    Before the first iteration, this is true because $\VCFunc(\VCVec) = \mathbf{0}$, so for every $U,R$, we have $\theta_R(U) \leq n/p^0 = n$, which always holds. Next, consider any iteration of the loop. Note that we can always pick a non-zero $\delta >0$, since for every subset $U$ that includes some variable outside of $H[\VCFunc(\VCVec)/\Delta]$, $\theta_R(U)< n/p^{\sum_{i\in U}\VCFunc_i(\psi)/\Delta}$. By construction, the weight of any heavy subset $U$ will not increase, so it will remain consistent. The statement is therefore true at the end of the iteration. 

    \textbf{(3)} We can write the norm of $\VCFunc(\VCVec)$ as follows: 
    \[
    |\VCFunc(\VCVec)|_1
    =\sum_{x\in V}\VCFunc_x(\VCVec)
    =\sum_{x\in V}\sum_{U\in\mathcal{P}(V)}\VCVec_U (\ve{v}^*_U)_x
    =\sum_{U\in\mathcal{P}(V)} \VCVec_U|\ve{v}^*_U|_1
    \leq \rqv(\mH) \cdot |\VCVec|_1
    \leq \rqv(\mH).
    \]
    Here the first inequality comes from the fact that $\rqv(\mH)$ is defined to be $\max_{U\subseteq V}|\ve{v}^*_U|_1$, and the second inequality from item {(1)}.
\end{proof}


\begin{lemma}\label{lem:guard}
 Let $Q$ be a $\Theta$-uniformized query. Let $\VCVec$ be the vector returned from Algorithm~\ref{algorithm:FindVertexWeights}. Then for every $R \in Q$ where $\vars{R} \setminus H[\VCFunc(\VCVec)/\Delta]\neq\emptyset$, there exists a relation $R' \in Q$ such that  $(i)$ $\vars{R} \setminus H[\VCFunc(\VCVec)/\Delta] \subseteq \vars{R'}$, and $(ii)$ $\sum_{x \in \vars{R'}} \VCFunc_x(\VCVec) \geq 1$.  
\end{lemma}

In this case, we say that $R'$ {\em guards} $R$. We also say that $R$ is {\em self-guarded} if it guards itself. In Example~\ref{example:runningexpickweights}, we will get that both $R$ and $T$ are guarded by $S$, and $S$ is self-guarded. Intuitively, our algorithm will use this concept as follows: a self-guarded relation can be directly partitioned with HyperCube with the desired load, whereas other relations, roughly speaking, are joined with their guard, creating an intermediate relation, which is partitioned with HyperCube.

\begin{proof}
First, we have $|\VCVec|_1 = 1$, since if $|\VCVec|_1 < 1$, then $H[\VCFunc(\VCVec)/\Delta]=V$, meaning $\vars{R} \setminus H[\VCFunc(\VCVec)/\Delta]=\emptyset$.

Let $t$ be the number of iterations before termination in Algorithm~\ref{algorithm:FindVertexWeights}. Let ${H}_0,\dots,{H}_t$ be the heavy sets throughout the algorithm, and $L_i = V \setminus H_i$. Observe that $L_0 \supset L_1 \supset \dots \supset L_t \supset \emptyset$. Let $\ve{v}^*_{L_1},\dots,\ve{v}^*_{L_t}$ be the minimum fractional vertex covers used in each iteration. Let $\delta_1,\dots,\delta_t$ be the quantities on Line $5$ in the algorithm for each vertex cover. Note that these correspond exactly to the non-zero elements in the vector $\psi$, hence $\sum_{j=1}^t \delta_j = 1$. 

Let $R,R'\in Q$. The following observation will be critical: if $\vars{R'} \cap L_j \supseteq \vars{R} \cap L_j$ then $\vars{R'} \cap L_i \supseteq \vars{R} \cap L_i$ for any $i \geq j$. Let us define the following: $R \prec R'$ if there exists $i=1, \dots, t$ such that $\vars{R} \cap L_i \subsetneq \vars{R'} \cap L_i$. It is easy to see that the relation $\prec$ is irreflexive, and it is also transitive (thus a strict preorder). Indeed, say $R \prec R'$ and $R' \prec R''$. Then, there exist $i,j$ such that $\vars{R} \cap L_i \subsetneq \vars{R'} \cap L_i$ and $\vars{R'} \cap L_j \subsetneq \vars{R''} \cap L_j$. Let $i \leq j$ w.l.o.g. Then, $\vars{R} \cap L_j \subseteq \vars{R'} \cap L_j \subsetneq \vars{R''} \cap L_j$, which implies $R \prec R''$.

Consider now any relation $R \in Q$. At the end of the algorithm, the light variables are $L_t$. Pick $R'$ to be a maximal relation w.r.t. $\prec$ such that $\vars{R'} \cap L_t \supseteq \vars{R} \cap L_t$. Such an $R'$ always exists, since we can pick $R'=R$, and $\prec$ is a strict preorder. $R'$ satisfies property $(i)$ of the lemma, so it suffices to show property $(ii)$: $\sum_{x \in \vars{R'}} \VCFunc_x(\VCVec) \geq 1$.

We claim that the following holds: for any $j=1, \dots, t$, 
$$ \sum_{x \in \vars{R'} \cap L_j} (\ve{v}^*_{L_j})_x \geq 1.$$

To show that the above holds, it suffices to prove that the hypergraph $\mH_j = \textsf{red}(\mH[L_j])$ has a hyperedge $e' = \vars{R'} \cap L_j$. Indeed, in this case the inequality would hold because $\ve{v}^*_{L_j}$ is a vertex cover for this hypergraph. The only case where $e'$ is not a hyperedge of $\mH_j$ is that $e'$ is strictly contained in another hyperedge $e''$. But that means that there exists a relation $R''$ such that $\vars{R'} \cap L_j \subsetneq \vars{R''} \cap L_j$, a contradiction to the fact that $R'$ is maximal w.r.t. $\prec$.

Finally, we can write the total weight on $R'$ as
    \begin{align*}
    \sum_{x\in \vars{R'}}\VCFunc_x(\VCVec) & =
    \sum_{x\in \vars{R'}}\sum_{j=1}^t \delta_j (\ve{v}^*_{L_j})_x = 
    \sum_{j=1}^t \delta_j \sum_{x\in \vars{R'}} (\ve{v}^*_{L_j})_x \\
    & = \sum_{j=1}^t \delta_j \sum_{x\in \vars{R'} \cap L_j} (\ve{v}^*_{L_j})_x
    \geq \sum_{j=1}^t \delta_j = 1.
    \end{align*} 
This completes the proof of the lemma.    
\end{proof}

    

\paragraph{Discussion} 
For certain instances, we may have $|\VCFunc(\psi)|_1<\Delta$ when picking $\Delta=\rqv$. This will cause the algorithm to use fewer than $p$ machines. It would then be possible to obtain a better load by picking a lower $\Delta$. If we pick $\Delta$ too low, we may get $|\VCFunc(\psi)|_1>\Delta$. By picking $\Delta=\rqv$, we achieve the best worst-case bound. Another, more fine-grained, method would be to minimize $\Delta$ by doing binary search, finding the $\Delta$ such that $|\VCFunc(\psi)|_1=\Delta$. This way, some instances, or partitions within instances, may obtain a better load than $n/p^{1/\rqv}$.

We will denote by $\ve{v}$ the vertex weight mapping $\VCFunc(\VCVec)/\Delta$, where $v_i=\VCFunc_i(\VCVec)/\Delta$.

\subsection{One Key Lemma}

Let us fix a $\Theta$-uniformized query $Q$ with constraint sets $\Theta = \{\theta_R\}_R$, and a consistent weight mapping $\ve{v}$ as chosen by the algorithm above. We define the {\em heavy relation}:
$$ R_H = \Join_{(U,R) \in \mathcal{U}} \pi_U (R).$$
The heavy relation $R_H$ has variables $H[\ve{v}]$. The heavy relation will be useful when partitioning relations which are not self-guarded. The next important lemma shows that $R_H$ behaves in a (somewhat) uniformized way.

\begin{lemma}\label{lem:heavy:uniform}
For any $X \subseteq H$, we have $\deg_{R_H} (H \mid X)\leq 2^{|H|} p^{\sum_{i\in H-X}v_i}$.
\end{lemma}

\begin{proof}
Consider any ordering $(U_1, R_1),\dots, (U_k, R_k)$ of the elements in the set $\mathcal{U}$. Let $H_0 = \emptyset$, and $H_j=\bigcup_{i=1}^j U_i$ for $j=1, \dots, k$. Note that $H_k = H$.
We will inductively show that for every $j=0, \dots, k$, we have $\deg_{R_H}(H_j \cup X \mid X ) \leq 2^j p^{\sum_{i\in H_j-X}v_i}$. 

In the base case $j=0$, the statement is true since $\deg_{R_H}(X \mid X ) = 1$. For the inductive case $j \geq 1$,
\begin{align*}
\deg_{R_H}(H_j \cup X |X) & \leq \deg_{R_H}(H_{j-1} \cup X | X)\deg_{R_j}(U_j|U_j\cap (H_{j-1}\cup X)) \\
&  \leq 2^{j-1} p^{\sum_{i\in H_{j-1}-X}v_i} \deg_{R_j}(U_j|U_j\cap (H_{j-1}\cup X)) &\text{(Inductive Assumption)} \\
&  \leq 2^{j} p^{\sum_{i\in H_{j-1}-X}v_i} \cdot \frac{\theta_{R_j}(U_j\cap (H_{j-1}\cup X))}{\theta_{R_j}(U_j)} &\text{(Def~\ref{def:sigmaconstraint})} \\
& \leq 2^{j}  p^{\sum_{i\in H_{j-1}-X}v_i} \frac{n}{p^{\sum_{i\in U_j\cap (H_{j-1}\cup X)}v_i}} \frac{p^{\sum_{i\in U_j}v_i}}{n} &\text{(Consistency, Def.~\ref{def:lightheavyconsistent})} \\
& = 2^{j}  p^{\sum_{i\in H_{j}-X}v_i}
\end{align*}
This concludes the proof.
\end{proof}

\subsection{The $\rqv$-Join Algorithm}

The algorithm has 4 phases. The first phase uses $O(1)$ rounds, and the remaining three phases need $4$ rounds. We first give some intuition behind the algorithm, and then describe it more precisely.

\paragraph{Intuition}
Because we picked $\ve{v}$ as a consistent weight mapping, there are no sets that are heavy with a strict inequality (see Definition~\ref{def:lightheavyconsistent}). Our main concern, and why we cannot apply HyperCube directly with the weight mapping, is that there may be relations that are not ``covered'' by $\ve{v}$, meaning the load will be too high.
To solve this problem, for such a relation $R$ that is not ``covered'', we can create an intermediate relation $R^\dagger$, by effectively joining $R$ with its guard $S$. The key idea is that this intermediate relation has a cardinality that grows very little compared to the input size, but it is guaranteed to be ``covered''. We can, after this, use HyperCube on all the obtained intermediate relations, which solves the query. From a technical perspective, we do not just join $R$ with its guard $S$, but instead we join the guard $S$ with the heavy relation $R_H$, followed by a semijoin with $R$. In fact, we may not need to use the whole relation $R_H$ for all $R$, but it never hurts to do so. We discuss this further later.

\paragraph{1-Preprocessing}
Partition the input using Algorithm~\ref{algorithm:UniformizePartition}. The remainder of the algorithm runs in parallel for each $\Theta$-uniformized query $Q$.

Use Algorithm~\ref{algorithm:FindVertexWeights} with $\Delta = \rqv(\mH)$ to find a consistent vertex weight mapping $\ve{v} = \VCFunc(\VCVec)/\Delta$.

\paragraph{2-Broadcasting Heavy Sets}

For every $(U,R) \in \mathcal{U}[{\ve{v}}]$, broadcast $\pi_UR$ to all machines. The heavy relation $R_H = \Join_{(U,R) \in \mathcal{U}} \pi_U (R)$ can now locally be computed in each machine.

\paragraph{3-Semijoins}
In parallel, for every relation $R \in Q$, do the following:
\begin{itemize}
    \item If $\vars{R} \setminus H[\ve{v}]\neq\emptyset$, let $S$ be the relation that guards $R$. Compute the join $R^\dagger \gets (S \Join R_H) \lJoin R$ using $p$ machines
    \item If $\vars{R} \setminus H[\ve{v}]=\emptyset$, compute the semijoin $R^\dagger \gets R_H\ltimes R$ using $p$ machines
\end{itemize}

\paragraph{4-HyperCube}
Perform HyperCube partitioning on every intermediate relation $R^{\dagger}$ using the share $p^{v_x}$ for variable $x$. 

\begin{example}
    We now continue Example~\ref{example:runningexpickweights}. The heavy sets are here $\mathcal{U}[\ve{v}]=\{(\{X\},R),(\{X\},T)\}$.
    
    Broadcasting Heavy Sets: We broadcast the heavy sets to all machines, after which each machine can compute the heavy relation $R_H$ (which only has variable $X$).

    Semijoins: We will compute two intermediate relations, both with variables $X,Y,Z$. These are $(S\bowtie R_H)\ltimes R$, and $(S\bowtie R_H)\ltimes T$. The intermediate results are at most $np^{1/6}$, since there are only $p^{1/6}$ values on $X$.

    HyperCube: We finally intersect the two intermediate relations with HyperCube partitioning.
\end{example}

We notice that the algorithm we get for the triangle query in the example is very similar to the one previously described \cite{oneroundwcoj}. Also, note that the algorithm above only uses $p^{5/6}<p$ machines, since $|\VCFunc(\psi)|_1=1.25$ and $\Delta=1.5$, so $|\ve{v}|_1=1.25/1.5=5/6$.

\paragraph{Discussion} Here, for some relation $R$, we use the relation $S\bowtie R_H$ and semijoin with $R$. Here, we could replace $R_H$ with $R_H'$, where $R_H'$ is different for each $R$, and only has a subset of variables in $H[\ve{v}]$. It suffices that $R_H'$ is a join of heavy sets such that $\vars{R_H'}\supseteq\vars{R}\cap H[\ve{v}]$. For binary relations, we can always pick $R_H'$ such that $\vars{R_H'}=\vars{R}\cap H[\ve{v}]$, which gives a nice interpretation that the semijoin phase effectively joins $R$ with its guard $S$, and semijoins on $R_H$, but for general arities, this may not be possible. Another option would be to directly partition self-guarded relations with HyperCube. We have chosen to simplify the algorithm by always joining with $R_H$ which has all heavy sets, allowing us to have fewer cases and to use the same relation $R_H$ for all relations $R$.

\subsection{Correctness and Load Analysis}

\paragraph{Correctness} To prove correctness of the above algorithm, we need to show two things. 

\begin{lemma}
    The total number of machines allocated is at most $p$.
\end{lemma}

\begin{proof}
To show this, it suffices to bound the number of machines used during the HyperCube phase. Indeed, the number of machines used to hash a relation $R^\dagger$ is 
$$p^{\sum_{x \in \vars{R^\dagger}}\ve{v}_x} \leq p^{|\ve{v}|_1} = p^{|\VCFunc(\VCVec)|_1/\Delta} = p^{|\VCFunc(\VCVec)|_1/\rqv(\mH)}  \leq p.$$ 
The last inequality follows from Lemma~\ref{lem:vector:properties}, item (3).
\end{proof}

\begin{restatable}{lemma}{algorithmcorrectness}
\label{lemma:algorithmcorrectness}
Let $Q$ be a $\Theta$-uniformized query. Then, the algorithm reports a tuple $t$ if and only if $t$ is in the output of $Q$.
\end{restatable}

\paragraph{Load Analysis} To analyze the load, we will need to use the following two results from prior literature. 

\begin{lemma}\label{lemma:HashPartitionLight}
    Consider a relation $R$ with arity $r$ and cardinality $ \leq N$. Suppose we wish to partition $R$ with $p$ machines, organized into a grid $[p_1]\times\dots\times[p_r]$ with $r$ hash functions. If for every $U\subseteq [r]$, we have $\deg_R(U) \leq \frac{N}{\prod_{i\in U}p_i}$, then with high probability the maximum load is $\tilde{O}(N/p)$.
\end{lemma}

\begin{lemma}\label{lem:semijoin}
The semijoin $R \lJoin S$ can be computed in two rounds with load $\tilde{O}(\max \{|R|,|S|\}/p)$.
\end{lemma}

Equipped with these two lemmas, we can show the crucial load result.

\begin{theorem}
Let $Q$ be a $\Theta$-uniformized query, and assume $n\geq p^3$. Then, the algorithm runs with load $\tilde{O}(n/p^{1/\rqv(\mH)})$.
\end{theorem}

\begin{proof}
We analyze the load across the remaining three phases. For simplicity, let $H = H[\ve{v}]$.

In the broadcasting phase, we broadcast all heavy sets. We know that for each $(U,R) \in \mathcal{U}[\ve{v}]$, $|\pi_U R| \leq p$. Each machine receives a constant number of heavy sets, each of size $p$, from at most $p$ machines. Hence, assuming $n \geq p^3$, this load is smaller than $n/p$.

For the semijoin phase, if $\vars{R} \setminus H=\emptyset$, then $R_H\ltimes R$ is performed with load $n/p$ by Lemma~\ref{lem:semijoin}. 

If $\vars{R} \setminus H\neq\emptyset$, we now bound $|S \Join R_H|$ for any relation $S$. Lemma~\ref{lem:heavy:uniform} tells us that $\deg_{R_H}(H \mid \vars{S} \cap H) = O(p^{\sum_{i \in H \setminus \vars{S}} v_i})$. Hence, $|S \Join R_H| = O(n p^{\sum_{i \in H \setminus \vars{S}} v_i})$. Using Lemma~\ref{lem:semijoin}, we obtain that the load for the semijoin operation is $\tilde{O}(n /p^{1-\sum_{i \in H \setminus \vars{S}} v_i})$. To show the required bound, it suffices to prove that $\sum_{x \in H \setminus \vars{S}} \VCFunc_x(\VCVec) \leq \rqv-1$. The key observation is that $S$ is a guard, and thus from Lemma~\ref{lem:guard}, $\sum_{x \in \vars{S}} \VCFunc_x(\VCVec) \geq 1$. Hence,
$$ \sum_{x \in H \setminus \vars{S}} \VCFunc_x(\VCVec) \leq \sum_{x} \VCFunc_x(\VCVec) - \sum_{x\in\vars{S}} \VCFunc_x(\VCVec) = \rqv - \sum_{x\in\vars{S}} \VCFunc_x(\VCVec) \leq \rqv-1.$$

Finally, we bound the load for the HyperCube phase. If $\vars{R} \setminus H=\emptyset$, then $|R^\dagger| \leq |R_H| \leq O(p^{\sum_{i\in H}v_i})\leq O(p)$ by Lemma~\ref{lem:heavy:uniform}. This means that no matter how $R^\dagger$ is partitioned by HyperCube, the load is lower than $n/p$.

If $\vars{R} \setminus H\neq\emptyset$, consider $R^\dagger$ from the semijoin phase. First, note that $|R^\dagger| \leq |S \Join R_H| = O(n p^{\sum_{x \in H \setminus \vars{S}} v_x})$ from Lemma~\ref{lem:heavy:uniform}. Second, $\sum_{x\in \vars{R^\dagger}} v_x = \sum_{x \in \vars{S}} v_x + \sum_{x \in H \setminus \vars{S}} v_x \geq 1/\rqv + \sum_{x \in H \setminus \vars{S}} v_x$.

Third, we will bound $\deg_{R^\dagger}(U)$ for every subset $U$. Indeed,
\begin{align*} 
\deg_{R^\dagger}(U) & \leq \deg_{S \Join R_H}(U) \leq \deg_{S}(U \cap \vars{S}) \cdot \deg_{R_H}(U \setminus \vars{S}) \\
& \leq \frac{n}{p^{\sum_{x \in U \cap \vars{S}} v_x}} \cdot p^{\sum_{x \in H \setminus (U \setminus \vars{S})} v_x} = \frac{n p^{\sum_{x \in H \setminus \vars{S}} v_x}}{p^{\sum_{x \in U} v_x}}
\end{align*}
We can now obtain the desired load by applying Lemma~\ref{lemma:HashPartitionLight} with $N = O(n p^{\sum_{x \in H \setminus \vars{S}} v_x})$.
\end{proof}

\subsection{Another Example}

\definecolor{varX}{HTML}{1F77B4}
\definecolor{varY}{HTML}{D62728}
\definecolor{varZ}{HTML}{000000}
\definecolor{varV}{HTML}{2CA02C}
\definecolor{varW}{HTML}{8E2BB8}
\newcommand{\colx}[1]{\textcolor{varX}{#1}}
\newcommand{\coly}[1]{\textcolor{varY}{#1}}
\newcommand{\colz}[1]{\textcolor{varZ}{#1}}
\newcommand{\colv}[1]{\textcolor{varV}{#1}}
\newcommand{\colw}[1]{\textcolor{varW}{#1}}

We finish the section by showcasing the algorithm on another query $Q$, depicted below. For the query below, we have $\kappa=5/2$, obtained by picking $S=\{x,y,z,v,w\}$ in Definition~\ref{def:kappa}. 
\[
S_1=\{\colx{x},\coly{y}\},
S_2=\{\coly{y},\colz{z}\},
S_3=\{\colx{x},\colz{z}\},
S_4=\{\colz{z},\colv{v}\},
S_5=\{\colz{z},\colw{w}\},
S_6=\{\colv{v},\colw{w}\}.
\]
\begin{figure}[H]
  \centering
    \begin{tikzpicture}[
    var/.style={circle, draw=#1, very thick, fill=#1!12, text=#1, minimum size=6mm, font=\small, inner sep=0pt},
    qedge/.style={line width=1.5pt, black!45},
    elabel/.style={font=\small\bfseries}
]
\node[var=varX] (x) at (-1.90,0.78)  {$x$};
\node[var=varY] (y) at (-1.90,-0.78) {$y$};
\node[var=varZ] (z) at (0,0)   {$z$};
\node[var=varV] (v) at (1.90,0.78)   {$v$};
\node[var=varW] (w) at (1.90,-0.78)  {$w$};
\draw[qedge] (x) -- node[elabel,left]           {$S_1$} (y);
\draw[qedge] (y) -- node[elabel,below,sloped]   {$S_2$} (z);
\draw[qedge] (x) -- node[elabel,above,sloped]   {$S_3$} (z);
\draw[qedge] (z) -- node[elabel,above,sloped]   {$S_4$} (v);
\draw[qedge] (z) -- node[elabel,below,sloped]   {$S_5$} (w);
\draw[qedge] (v) -- node[elabel,right]          {$S_6$} (w);
\end{tikzpicture}
  \caption{Visualization of the query $Q$}
\end{figure}
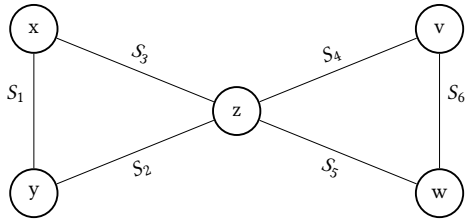

We consider the following $\Theta$-uniformized instance: $S_5$ is a product instance with $p^{1/25}$ values on $z$ and $n/p^{1/25}$ values on $w$. $S_6$ is a product instance with $p^{3/25}$ values on $w$ and $n/p^{3/25}$ values on $v$. All other relations are arbitrary with domain $n$ and degrees below $n/p$.

We will pick weights for our HyperCube algorithm using a linear combination of vertex covers, as described in Algorithm~\ref{algorithm:FindVertexWeights}. Denote a vertex cover as a vector $(v_{\colx{x}},v_{\coly{y}},v_{\colz{z}},v_{\colv{v}},v_{\colw{w}})$. In the first iteration, we use $\ve{v}_{XYZVW}=(\colx{0.5},\coly{0.5},\colz{0.5},\colv{0.5},\colw{0.5})$ and get $\psi_{XYZVW}=1/5$, since then $\colz{z}$ has weight $\frac{1}{5}\cdot\frac{1}{2}=\frac{1}{10}$ and becomes heavy in $S_5$ ($n/p^{1/25}=n/p^{(1/10)/\kappa}$). In the second iteration, we use $\ve{v}_{XYVW}=(\colx{0.5},\coly{0.5},\colz{0},\colv{0.5},\colw{0.5})$ and get $\psi_{XYVW}=2/5$, since then $\colw{w}$ has weight $\frac{1}{10}+\frac{2}{5}\cdot\frac{1}{2}=\frac{3}{10}$ and becomes heavy in $S_6$ ($n/p^{3/25}=n/p^{(3/10)/\kappa}$). In the final iteration, we use $\ve{v}_{XYV}=(\colx{0.5},\coly{0.5},\colz{0},\colv{1},\colw{0})$ and get $\psi_{XYV}=2/5$. In total, we get the following weights

\[
\ve{v}
=\frac{1}{5}\left(\colx{\tfrac{1}{2}},\coly{\tfrac{1}{2}},\colz{\tfrac{1}{2}},\colv{\tfrac{1}{2}},\colw{\tfrac{1}{2}}\right)+\frac{2}{5}\left(\colx{\tfrac{1}{2}}, \coly{\tfrac{1}{2}}, \colz{0}, \colv{\tfrac{1}{2}}, \colw{\tfrac{1}{2}}\right) + \frac{2}{5}\left(\colx{\tfrac{1}{2}}, \coly{\tfrac{1}{2}}, \colz{0}, \colv{1}, \colw{0}\right)
=\frac{1}{10}(\colx{5},\coly{5},\colz{1},\colv{7},\colw{3}).
\]

Note that $|\ve{v}|_1=\frac{21}{10}<\kappa$. We can now use $\ve{v}$ to compute $Q$ with $p^{|\ve{v}|_1/\kappa}=p^{21/25}$ machines with load $n/p^{2/5}$.

\paragraph{Broadcasting Heavy Sets} There are two heavy sets in our instance, with respect to $\ve{v}$: $(\{z\},S_5)$ and $(\{w\},S_6)$. Broadcast $\pi_zS_5$ and $\pi_wS_6$ to all machines. Each machine can now compute $R_H$ by taking the Cartesian product of these heavy sets.

\paragraph{Semijoin Phase} In this query, $S_1,S_6$ are self-guarded. $S_2$ and $S_3$ are both guarded by $S_1$, and $S_4$ is guarded by $S_6$. $S_4$ is not self-guarded since $v_{\colz{z}}+v_{\colv{v}}=\frac{1}{10}+\frac{7}{10}<1$, whereas $S_6$ has $v_{\colv{v}}+v_{\colw{w}}=1$. Each machine locally joins all tuples with $R_H$. Then, for any $R\in Q\setminus \{S_5\}$, with guard $S$, we compute the semijoin $(S\bowtie R_H)\ltimes R$. For $S_5$, we just compute $R_H\ltimes S_5$.

\paragraph{HyperCube Phase} 
Allocate a share of $p^{v_x/\kappa}$ machines to variable $x$. Then run HyperCube on all intermediate relations from the previous phase.

\paragraph{Discussion} As discussed previously, we have simplified our algorithm by joining every relation with $R_H$ before doing HyperCube. In the example above, we could have instead done HyperCube with the following relations: $S_1$, $S_5\ltimes\pi_wS_6$, $S_6$, $(S_1\bowtie \pi_zS_5)\ltimes S_2$, $(S_1\bowtie \pi_zS_5)\ltimes S_3$, $(S_6\bowtie S_5)\ltimes S_4$.

In \appref{sec:dedicatedmachines:app}, we give another example, the Loomis--Whitney query.

\section{Towards A Tight Lower Bound}
\label{sec:lowerbounds}

In this section, we study the optimality of the $\tilde{O}(n/p^{1/\rqv})$ load with respect to lower bounds. 

We will restrict our attention to the class of {\em tuple-based MPC algorithms}, where we have to transmit tuples in their entirety. Tuples can be copied in each machine, and a machine is only allowed to output a tuple $t$ if it has received all the tuples of the form $t[\vars{R}]$, for every $R \in Q$. The restriction that only input tuples can be transmitted is not material: we can loosen this restriction and also allow transmitting intermediate tuples. These are the same within a constant factor in load, since when transmitting an intermediate tuple, we can always just transmit a set of input tuples that witness the intermediate tuple. We also note that all known MPC algorithms for joins (including $\rqv$-Join) can be captured by the tuple-based variant. 

Prior work~\cite{oneroundwcoj} has shown a straightforward $\Omega(n/p^{1/\rho^*})$ lower bound for any join query. This immediately implies the following corollary.

\begin{proposition}
Let $Q$ be a query with hypergraph $\mH$ such that $\rqv(\mH) = \rho^*(\mH)$. Then, the $\rqv$-Join algorithm is optimal for $Q$ within a polylogarithmic factor in the tuple-based model.
\end{proposition}

Two families of hypergraphs that have the above property are binary and acyclic hypergraphs (see Lemmas~\ref{lem:kapparhorange} and~\ref{lem:acyclic}). For Loomis--Whitney joins with $k$ variables, their hypergraphs also satisfy the above property: $\rqv = \rho^* = k/(k-1)$. Below, we show a hypergraph from~\cite{aamerketsmanpac} where $\rqv$-Join achieves optimality and, to the best of our knowledge, no such algorithm was known before.

\begin{example}
Consider the hypergraph $\mathcal{G}$ with the following hyperedges:
$$\{a,b\}, \{b,c,d,e\}, \{b,e,f\}, \{e,f,g\},\{g,h\},\{g,i\}, \{h,i\}. $$
For this example, we have $\rho^*(\mathcal{G}) = \rqv(\mathcal{G}) = 4$ (the hypergraph $\textsf{red}(\mathcal{G}[a,c,d,f,h,i])$ has $\tau^*=4$). Thus, the above proposition gives an optimal load of $\tilde{O}(n/p^{1/4})$.
\end{example}


A more involved lower bound that is higher than $\rho^*$  was given in~\cite{hutaoacyclic} for the class of boat queries $\mathcal{H}_k^B$, $k\geq 2$, with hyperedges:
$ \{ \{ x_1,x_2,\dots, x_k\}, \{y_1,y_2,\dots, y_k\}, \{x_1, y_1 \}, \{x_2, y_2\}, \dots, \{x_k, y_k\} \}. $
For this class of hypergraphs, the lower bound shown is $\Omega(n/p^{1/k})$. Note that $\rqv(\mathcal{H}_k^B) = \psi^*(\mathcal{H}_k^B) = \tau^*(\mathcal{H}_k^B) = k$ (but $\rho^*=2$), hence the lower bound can be matched even by a 1-round MPC algorithm.

We show next how to extend the above lower bound to the family of hypergraphs $\mathcal{H}_k^\dagger$. For this, we need the following general observation.

\begin{restatable}{lemma}{reducedlowboundtransfer}
\label{lem:lb:reduced}
Let $\mH = (V,E)$ be a hypergraph, and $S \subseteq V$. If $\textsf{red}(\mH[S])$ admits a lower bound of $\Omega(n/p^{\epsilon})$, then $\mH$ also admits the same lower bound.
\end{restatable}

Observe that $\mathcal{H}_k^B = \mathcal{H}_k^\dagger[V \setminus \{z\}]$. Thus, the $\Omega(n/p^{1/k})$ lower bound transfers to $\mathcal{H}_k^\dagger$, showing that $\rqv$-Join is optimal for this family of queries as well. Together with Proposition~\ref{boatzqueryquantities}, this proves that the worst-case optimal load for join queries cannot be characterized as $n/p^{1/\max\{\tau^*,\rho^*\}}$.



\paragraph{Candidate Lower Bound}
We now describe a concrete direction towards constructing a matching lower bound. Recall that $\rqv(\mH) = \max_{S \subseteq V} \tau^*(\textsf{red}(\mH[S]))$. Applying Lemma~\ref{lem:lb:reduced}, it would suffice to prove an $\Omega(n/p^{1/\tau^*})$ lower bound for every reduced hypergraph.

\begin{conjecture}
    Let $\mH$ be a hypergraph s.t. $\mH=\textsf{red}(\mH)$. Then, there exists a query $Q$ with hypergraph $\mH$ and input size $n$ s.t. any tuple-based algorithm needs load $\Omega(n/p^{1/\tau^*})$ to compute $Q$.
\end{conjecture}

Proving this conjecture would immediately give a matching $\Omega(n/p^{1/\rqv})$ lower bound. We should note here that the requirement that $\mH$ is reduced is critical; for instance, the conjecture fails for the non-reduced hypergraph with hyperedges $\{x\}, \{x,y\}, \{y\}$, for which there is an MPC algorithm with load $\tilde{O}(n/p)$.


Although we could not prove the conjecture, we describe below a candidate family of hard instances. For a query with (reduced) hypergraph $\mH = (V,E)$, let $\ve{v}$ be a minimum fractional vertex cover, which has total weight $\tau^*$. A \textit{sparse product query} is defined as follows:
\begin{itemize}
    \item For a variable $x \in V$, define $N_x = \{1, \dots, n^{v_x}\}$.
    \item For relation $R$, include each tuple in the cross product $\bigtimes_{x \in \vars{R}} N_x$ independently with probability 
    $
    P_R = n^{1-\sum_{x\in \vars{R}}v_x}.
    $
\end{itemize}

We call a relation $R$ \textit{sparse} if $P_R<1$, and \textit{dense} if $P_R=1$. Note that with high probability, every relation has size $O(n)$. 
Additionally, the expected output of the query is $n^{\tau^*} \prod_{R\in Q}P_R$. Finally, if we consider a machine that receives $L$ tuples at most from each relation, then we show in \appref{sec:lb:appendix} that the expected number of outputs it can produce is at most $L^{\tau^*} \prod_{R\in Q}P_R$. 

The above in-expectation bounds do not suffice to prove the desired lower bound. Instead, we would want them to hold with sufficiently high probability such that there is a non-zero probability of the existence of an instance of sparse relations such that for every possible combination of $L$ tuples a machine can produce $O(L^{\tau^*} \prod_{R\in Q}P_R)$ results. We leave this as an open problem, which, if resolved, would prove the desired matching lower bound.

\section{Conclusion}

In this paper, we presented a new algorithm, $\rqv$-Join, for computing join queries in the MPC model with load $n/p^{1/\kappa}$. The algorithm improves upon previous results and is relatively simple. It remains an open problem whether this upper bound is tight, or if there exists another algorithm that can achieve a better load.

\newpage
\bibliographystyle{ACM-Reference-Format}
\bibliography{references}

\newpage
\appendix
\section{$\kappa$ as a Mixed Integer Linear Program}
\label{sec:kappamilp}

Consider a hypergraph $\mH=(V,E)$. Consider the mixed integer linear program below, where we have a variable $w_R\in[0,1]$ for each $R\in E$, and a variable $t_x\in\{0,1\}$ for each $x\in V$. The variable $t_x$ indicates whether variable $x$ is in the set $S$, so by setting values of $t_x$, we pick a set $S$. Here, $L$ is some large number.

\begin{equation}\label{eq:MixIntLinProgKappa}
\begin{aligned}
    \max        \quad & \sum_{R\in E}w_R                                  &&                                     \\
    \text{s.t.} \quad & \sum_{R:x\in R} w_R \leq t_x+(1-t_x)L,            && \forall x\in V                       \\
                      & w_R \leq \sum_{x\in R}t_x,                        && \forall R\in E                       \\
                      & w_W \leq 1 + \sum_{x\in W-U}t_x - t_{y},          && \forall U,W\in E,\ \forall y\in U-W   \\
                      & t_x \in\{0,1\},                                   && \forall x\in V                       \\
                      & w_R \in[0,1],                                     && \forall R\in E
\end{aligned}
\end{equation}

\begin{lemma}
\label{lemma:settvarsmilp}
    Consider a hypergraph $\mH=(V,E)$, and let $S$ be the variables where $t_x=1$ in (\ref{eq:MixIntLinProgKappa}). Then the optimal objective to (\ref{eq:MixIntLinProgKappa}) is $\tau^*(\textsf{red}(\mathcal{H}[S]))$.
\end{lemma}
\begin{proof}
    The first constraint is the edge packing constraint on only the variables that are in $S$.
    The second constraint describes that only hyperedges with at least one variable in $S$ can have a non-zero weight.
    It remains to show that the third constraint guarantees that a hyperedge that contains a variable in $S$, and is removed by the reduce operation,
    never has a non-zero weight.

    Suppose that $W$ is a hyperedge that is removed because it is contained in $U$. If $\pi_SU=\pi_SW$, both $U$ and $W$ are kept in the program,
    which is correct, since they create the same constraints. We focus on the case when $(U-W)\cap S\neq\emptyset$.

    In this case, there exists a constraint $w_W\leq 0$ among the third type of constraints. This is because $\sum_{x\in W-U}t_x=0$ since $W$ is reduced into $U$, and there exists some $y\in U-W$ where $y\in S$, meaning $t_y=1$, causing the RHS to be $0$.

    If $W$ is not reduced into $U$, then RHS of the third constraint is at least $1$, since $\sum_{x\in W-U}t_x\geq1$.
\end{proof}

\begin{theorem}
    Consider a hypergraph $\mH=(V,E)$. Then, the optimal objective of (\ref{eq:MixIntLinProgKappa}) w.r.t. $\mH$ is $\kappa(\mH)$.
\end{theorem}
\begin{proof}
    This follows since (\ref{eq:MixIntLinProgKappa}) maximizes over all choices on variables $\{t_x\}_{x\in V}$, and applying Lemma~\ref{lemma:settvarsmilp}.
\end{proof}

\iffull
\section{Allocating Dedicated Machines}
\label{sec:dedicatedmachines:app}

We here give an example on the Loomis--Whitney query, which showcases a limitation of the approach of directly allocating to each heavy tuple a dedicated set of machines.

Consider the Loomis--Whitney join, whose hypergraph $\mH_{LW}$ has the following hyperedges
\[
S_1=\{x,y,z\},S_2=\{x,y,w\},S_3=\{x,z,w\},S_4=\{y,z,w\}.
\]
We have the following quantities. Here, the last quantity is the characterization of the load in \cite{qiaotao2attribute}, where $\alpha$ is the maximum arity and $\phi$ is the generalized vertex packing.
\begin{itemize}
    \item $\rho^*(\mH_{LW})=4/3$
    \item $\tau^*(\mH_{LW})=4/3$
    \item $\psi^*(\mH_{LW})=2$
    \item $\rqv(\mH_{LW})=4/3$
    \item $\frac{\alpha\phi}{2}=\frac{3\cdot \frac{4}{3}}{2}=2$
\end{itemize}

\paragraph{Instance}
Consider the following $\Theta$-uniformized instance.
Let $S_3$ (and $S_4$) be product relations where $|\pi_{zw}S_3|=|\pi_{zw}S_4|=n/p^{1/5}$, and $|\pi_x S_3| = |\pi_y S_4| = p^{1/5}$. Let $S_1$ (and $S_2$) be product relations where $|\pi_{xy}S_1|=|\pi_{xy}S_2|=p^{2/5}$, and $|\pi_zS_1|=|\pi_wS_2|=n/p^{2/5}$.

We cannot solve this with HyperCube in one round, because if we use the vertex cover $\left(\frac{1}{3},\frac{1}{3},\frac{1}{3},\frac{1}{3}\right)$, $\deg_{S_3}(x)=n/p^{1/5}>n/p^{(1/3)/(4/3)}=n/p^{1/4}$.

\paragraph{Dedicated machines}
Consider an approach that allocates dedicated machines to each heavy tuple. There are $p^{2/5}$ heavy tuples, so some heavy tuple will get at most $p^{3/5}$ machines. If this group computes the semijoin $S_i\ltimes t$ by receiving all of $S_i$, this requires load $n/p^{3/5}$, and we fail to match $\rqv$.

Intuitively, the problem is that a relation of size $n$ is too big to be received by a dedicated group. The groups need to ``collaborate'' in partitioning any input relation.

\paragraph{$\kappa$-Join}
We can compute this query on this $\Theta$-uniformized instance with load $n/p^{4/5}$ by running Algorithm~\ref{algorithm:FindVertexWeights} with $\Delta=5/4$. This gives the following weight mapping, where variable $x$ gets the share $p^{v_x/\Delta}$:
\[
\ve{v}=\left(\frac{2}{8},\frac{2}{8},\frac{3}{8},\frac{3}{8}\right) 
= \frac{3}{4}\left(\frac{1}{3},\frac{1}{3},\frac{1}{3},\frac{1}{3}\right) + \frac{1}{4}\left(0,0,\frac{1}{2},\frac{1}{2}\right)
\]

We can see that $S_3$ (and $S_4$) are covered by this weight mapping, and even though $x$ (or $y$) is heavy in these relations, it is consistent (both the degree and the threshold are $n/p^{1/5}$) so we could partition with HyperCube.

For $S_1$ (and $S_2$), we cannot do this, since they are not self-guarded. If we instead perform a join $(S_3\bowtie R_H)\ltimes S_1$ (and $(S_4\bowtie R_H)\ltimes S_2$) the output size is at most $np^{1/5}$. These intermediates can then be repartitioned using HyperCube with load $np^{1/5}/p=n/p^{4/5}$.

\fi

\section{Missing Proofs}

\subsection{Section~\ref{sec:measure}}

\kappacomparison*
\begin{proof}
From Lemma~\ref{lemma:KappaCCSum}, we can assume w.l.o.g. that $\mH$ is connected. \\
    \textbf{1)} This is trivial, since we can pick $S=V$ in the definition of $\rqv$. \\
    \textbf{2)} 
    First we show that $\rho^*(\textsf{red}(\mH)) = \rho^*(\mH)$. Indeed, given a fractional edge cover of $\textsf{red}(\mH)$, we can obtain one for $\mH$ with the same weight by assigning $0$ weight to the hyperedges removed. For the opposite direction, suppose we have a fractional edge cover $\ve{u}$ of $\mH$. For every edge $e'$ removed to form  $\textsf{red}(\mH)$, let $f(e')$ be an edge in $\textsf{red}(\mH)$ such that $e' \subseteq f(e')$. We now create an edge cover for $\textsf{red}(\mH)$ by assigning to $e$ weight $u_e + \sum_{e': f(e')=e} u_{e'}$. The total weight remains the same, since the weight of every removed edge is assigned to one edge. It is also an edge cover for $\textsf{red}(\mH)$, since the weight of a removed edge is assigned to an edge that covers at least the same vertices.
    
    For a fractional edge cover $\ve{u} = (u_j)_j$ of $\mH = (V,E)$, let $S_{\ve{u}} = \{ v \in V \mid \sum_{v \in e} u_e = 1 \}$, i.e. the vertices with a tight constraint in the edge cover.
    Let $\ve{w}$ be an optimal fractional edge cover of $\mH$ such that $S_{\ve{w}}$ is a minimal set (i.e., there is no other optimal fractional edge cover with a strictly smaller number of tight vertices). We will show that $\tau^*(\textsf{red}(\mH[S_{\ve{w}}])) \geq \rho^*(\mH)$. This proves the required statement.

    To show this, consider an edge $e$ that is removed to form $\textsf{red}(\mH[S_{\ve{w}}])$ s.t. $e \cap S_{\ve{w}}  = \emptyset$. This means that for every $v \in e$, we have that $\sum_{ v \in e'} w_{e'} > 1$. If $w_e>0$, then it would be possible to construct an edge cover $\ve{w}'$ of smaller weight by reducing ${w}_e$ by some $\epsilon>0$, a contradiction to the optimality of $\ve{w}$. Hence, any such edge has weight 0.

    We will construct an edge packing $\ve{c}$ for $\textsf{red}(\mH[S_{\ve{w}}])$ by assigning to edge $e$ weight $c_e = w_e + \sum_{e': f(e')=e} w_{e'}$. As before, the total weight remains the same. Hence, it suffices to show that it is indeed an edge packing for $\textsf{red}(\mH[S_{\ve{w}}])$. Suppose not; then, there is a vertex $v$ that was tight under $\ve{w}$ but $\sum_{ v \in e} c_e >1$. This means that there is an edge $e$ with $w_e>0$ and $v \notin e$ such that it is absorbed by at least another edge $e'$ with $v \in e'$. Moreover, from our above observation and because $w_e >0$, $e$ must contain at least one tight vertex $v'$ that also belongs to $e'$.
    Now, we can transfer a small $\epsilon >0$ from $w_e$ to $w_{e'}$; but this means that we will get an optimal edge cover where $v$ becomes non-tight, a contradiction to the minimality of $\ve{w}$.\\
\textbf{3)}  Recall that $\rqv(\mH) :=\max_{S \subseteq V}\tau^*(\textsf{red}(\mH[S]))$, and $\psi^*(\mH) :=\max_{S \subseteq V}\tau^*(\mH[S])$. It suffices to show that for every $S \subseteq V$, $\tau^*(\mH[S])\geq\tau^*(\textsf{red}(\mH[S]))$. This is true, since viewing these two as linear programs with minimization objectives, moving from $\mH[S]$ to $\textsf{red}(\mH[S])$ removes constraints, which can only lower the value of the objective function.
\end{proof}

\acyclickapparho*
\begin{proof}
Assume w.l.o.g. that $\mH$ is reduced. We need to show that for every $U \subseteq V$, we have $\tau^*(\textsf{red}(\mH[U])) \leq \rho^*(\mH)$. We will use the notion of a {\em canonical edge cover} from~\cite{hutaoacyclic}, which is integral. For this, take a join tree $T$ of $\mH$ (such a tree exists because $\mH$ is acyclic) and root it at some node. Each node $v$ of $T$ is associated with a hyperedge of $\mH$, which we denote by $\chi(v)$. We will say that a variable is {\em disappearing} at a node $v$ of $T$ if it does not appear in $\chi(u)$ for any node $u$ above $v$. We can now define the canonical cover as follows: we construct a set $\mathcal{F}$ of nodes of $T$, starting with $\mathcal{F} = \emptyset$, by visiting the nodes of $T$ in some bottom-up order. Whenever we visit a node $v$ of $T$ such that $\chi(v)$ contains a disappearing variable that is not in $\bigcup_{u \in \mathcal{F}} \chi(u)$, add $v$ to $\mathcal{F}$. The resulting set $\{\chi(v) \mid v \in \mathcal{F}\}$ can be shown to be an edge cover~\cite{hutaoacyclic}, and its value will be equal to $\rho^*(\mH)$. 

We will now construct a vertex cover for $\textsf{red}(\mH[U])$ from $\mathcal{F}$ that has total weight $ \leq |\mathcal{F}|$. We do this as follows. For every node $v \in \mathcal{F}$ of the join tree, assign a weight of 1 to the variable in $\chi(v) \cap U$ that disappears as high as possible in the rooted join tree; for the remaining variables, assign a weight of 0. It is easy to see that the total weight we have is at most $|\mathcal{F}|$. It remains to show that every edge in $\textsf{red}(\mH[U])$ is covered. 

Indeed, consider any node $v$ of $T$ such that $\chi(v) \cap U$ is an edge of $\textsf{red}(\mH[U])$. If $v \in \mathcal{F}$, we are done, since we will pick one of the variables in $\chi(v) \cap U$ to have weight 1. If $v \notin \mathcal{F}$, note that there must be a variable $x \in \chi(v) \cap U$ that is disappearing at $v$; otherwise, every variable in $\chi(v) \cap U$ would also appear in $\chi(w)$, where $w$ is the parent of $v$, which would mean that $\chi(v) \cap U$ does not occur in the reduced hypergraph $\textsf{red}(\mH[U])$. Since $v \notin \mathcal{F}$, this means that $x$ must occur in $\chi(u)$ for some node $u \in \mathcal{F}$ below $v$. Now, for $u$ we chose to put weight of 1 on some variable $y$ that disappears at least as high as $x$; this means that $y$ must also appear in $\chi(v)$, thus $\chi(v) \cap U$ is covered.
\end{proof}

\kapparhorange*
\begin{proof}
    The first inequality is proven in Lemma~\ref{lemma:kappacomparison}. It remains to show $\rqv \leq \frac{\alpha}{2}\rho^*(\mH)$ when $\alpha \geq 2$. Equivalently, we show that for every $S \subseteq V$, $\tau^*(\textsf{red}(\mH[S])) \leq \frac{\alpha}{2}\rho^*(\mH)$. Here, in turn, it suffices to prove that there exists a fractional vertex cover $\hat{\tau}$ over $\textsf{red}(\mH[S])$ and a fractional vertex packing $\hat{\rho}$ over $\mH$ such that $\hat{\tau} \leq \frac{\alpha}{2}\hat{\rho}$, since then
    \[
    \tau^*(\textsf{red}(\mH[S])) \leq \hat{\tau} \leq \frac{\alpha}{2}\hat{\rho} \leq \frac{\alpha}{2}\rho^*(\mH)
    \]
    Here, the first and last inequalities hold since $\hat{\tau}$ (and $\hat{\rho}$) can only be bigger (resp. smaller) than the optimal values.
    We say that a vertex $v$ in the query $\textsf{red}(\mH[S])$ has a neighbor if there is a hyperedge $E\in \textsf{red}(\mH[S])$ where $v\in E$ and there exists some $v'\in E$ where $v\neq v'$. Partition vertices in $\textsf{red}(\mH[S])$ into those that have a neighbor and those that do not.
    We create the fractional vertex cover $\hat{\tau}$ by assigning weight $1/2$ to every vertex with a neighbor, and weight $1$ to every vertex without a neighbor. Clearly, this creates a vertex cover.
    We can now create the fractional vertex packing $\hat{\rho}$ by using $\hat{\tau}$ and dividing each assigned weight by $\frac{\alpha}{2}$. This is a vertex packing since every edge $E$ containing only one $v\in S$ has total weight $\frac{2}{\alpha}\leq 1$ (as $\alpha \geq 2$), and an edge $R$ with $r\geq 2$ such vertices has total weight $\frac{r}{\alpha} \leq 1$. Clearly $\hat{\tau}=\frac{\alpha}{2}\hat{\rho}$.
\end{proof}

\boatzqueryquantities*
\begin{proof}
We prove each item of the proposition. \\
    \textbf{1)} The minimum edge cover assigns weight $1$ on $\{x_1,\dots, x_k,z\}, \{y_1,\dots,y_k,z\}$. \\
    \textbf{2)} The minimum vertex cover assigns weight $1$ on $z$, and $0$ to every other variable. \\
    \textbf{3)} Consider the subset $S=\{x_1,\dots,x_k,y_1,\dots,y_k\}$ that only leaves out $z$. Here, the minimum vertex cover for $\textsf{red}(\mH_k^\dagger[S])$ is $k$.
\end{proof}

\subsection{Section~\ref{sec:upperbound}}

\algorithmcorrectness*
\begin{proof}
    \textbf{$\Rightarrow$)} This follows since our algorithm is tuple-based, which means a machine only reports $t$ if it stores witnessing input tuples for $t$.
    
    \textbf{$\Leftarrow$)} Since we know that the HyperCube algorithm is correct, to prove the other direction, it suffices to show that $\Join_i R_i = \Join_i R_i^\dagger$.
    \begin{enumerate}
        \item $\Join_i R_i \supseteq \Join_i R_i^\dagger$: Let $t^\dagger$ be a tuple in $\Join_i R_i^\dagger$. For every $R_i^\dagger$, there exists a $t_i^\dagger$ that witnesses $t$. Then $R_i$ has a tuple $t_i=\pi_{\vars{R_i}}t_i^\dagger$. These tuples cause $t^\dagger$ to also be in $\Join_i R_i$.
        \item $\Join_i R_i \subseteq \Join_i R_i^\dagger$: Let $t$ be a tuple in $\Join_i R_i$. Then, there exist a tuple $t_i$ in $R_i$ that witnesses $t$, a tuple $\hat{t}_i$ in the guard $S_i$ of $R_i$ (if $R_i$ has a guard) that witnesses $t$, and $t_H=\pi_{H[\ve{v}]}t$. $R_i^\dagger$ then has a tuple $t_i^\dagger$, obtained with $t_i,\hat{t}_i,t_H$ (or $t_i,t_H$ if there is no guard), such that $t_i^\dagger$ witnesses $t$, for each relation $R_i^\dagger$, causing $t$ to also be in $\Join_iR_i^\dagger$.\end{enumerate}
\end{proof}

\subsection{Section~\ref{sec:lowerbounds}}

\reducedlowboundtransfer*
\begin{proof}
Consider the query $Q'$ with hypergraph $\textsf{red}(\mH[S])$ that achieves the lower bound of $\Omega(n/p^{\epsilon})$. We can create a query $Q$ for $\mH$ as follows. Every relation $R \in Q$ that is not removed in $\textsf{red}(\mH[S])$ is the same as in $Q'$, with the only difference that any attribute in $\vars{R} \setminus S$ has a fixed constant value (of say, 0). If a relation $R \in Q$ is removed in $\textsf{red}(\mH[S])$, then it is absorbed by some relation $R' \in Q'$; in this case, $R$ is the same as the projection of $R'$ on $\vars{R} \setminus S$, setting any attribute in $\vars{R} \setminus S$ to 0 as before. For query $Q$, all relations have size at most $n$, as in $Q'$. It is now easy to see that any tuple-based algorithm that correctly computes $Q$ can be modified to compute $Q'$ with exactly the same load. 
\end{proof}

\section{Comparison of $\kappa$ and PAC}
\label{sec:comparepac:app}

In this section, we will compare $\rqv$ to $\PAC$ \cite{aamerketsmanpac}. 

\comparisonkappapac*
\begin{proof}
Consider any query $Q$ over $\mathcal{H}$, and let $S\subseteq V$ be the subset of vertices that maximizes $\tau^*(\textsf{red}(\mathcal{H}[S]))$. Consider a database instance where all variables outside $S$ only have one value $1$, and on variables in $S$, we have values $1,\dots,n$. Each relation will have all tuples that have value $1$ on attributes outside $S$, and the same value $s\in\{1,\dots, n\}$ on all variables in $S$. For example, for a relation on attributes $x,y,z$ where $x\notin S$ and $y,z\in S$, we will have tuples $(1,1,1),(1,2,2),(1,3,3),\dots,(1,n,n)$. Note that the query over this instance can be computed with load $n/p$, by performing repeated binary joins, since the intermediate sizes do not grow, but we will argue that under the PAC framework, any solution incurs load $n/p^{1/\tau^*(\textsf{red}(\mathcal{H}[S]))}$, which matches $\kappa$.

Note that for this instance, any set $A$ of variables where $A\subseteq S$, has the configuration $C(A)=1$, meaning no $A\subseteq S$ is heavy.

Consider the variables in $S$. These can never be in a patch, since they are not heavy. This leaves us with covering $S$ with a semi-cover and anchors.

We will argue that the total cost of these is at least $\tau^*(\textsf{red}(\mathcal{H}[S]))$. Indeed, consider the following vertex weight mapping $v$: each vertex in $S$ gets the weight $f_x$ from the semi-cover. In addition, for an anchor on relation $R$, we add a total weight $1$ to the vertices in $R\cap S$. \PAC{} requires that every relation $R$ that has some variable in $S$ is either (1) covered by $f$, (2) reducible into a relation that is covered by $f$ or an anchor or (3) is an anchor. In cases 1 and 3, this means that $v$ covers $R$. In case $2$, $R$ would be left out from $\textsf{red}(\mathcal{H}[S])$. Hence, $v$ is a vertex cover over $\textsf{red}(\mathcal{H}[S])$. Since every vertex cover has weight at least $\tau^*(\textsf{red}(\mathcal{H}[S]))$, the lemma follows.
\end{proof}

\section{Further comparison of $\kappa$ with query classes from PAC}
\label{sec:compareclasspacv2}

In a paper in ICDT 2026 \cite{pacv2}, the authors of \PAC{} present upper bounds for some specific query classes. These upper bounds build on the PAC algorithm, but make some modifications. In this section, we will show that even with these modifications, the upper bounds are matched by $\kappa$.

\paragraph{Degree-two queries}

\begin{definition}
    A \textit{linear degree 2 query} is a query where every variable is in exactly two hyperedges, and every pair of hyperedges share at most one variable.
\end{definition}

\begin{lemma}
\label{lemma:degree2kappa}
    Consider a hypergraph $\mH$ which is linear and degree-2, where $\mH = \textsf{red}(\mH)$.
    Then, $\kappa \leq \min\{ k/2, \psi^*\}$, where $k$ is the number of variables and $\psi^*$ is the edge quasi-packing value.
\end{lemma}
\begin{proof}
    Since $\kappa\leq\psi^*$ holds for any query, it suffices to show $\kappa\leq k/2$.

    First, $\mH$ has no unary hyperedge $\{x\}$: $x$ lies in exactly two hyperedges $E_1,E_2$. If $E_1$ is unary, it will be removed by reduction into $E_2$ (even if $E_1=E_2$). But then the hypergraph is not reduced and linear degree-2.

    Fix $S\subseteq V$. Call $x\in S$ \emph{isolated} if both hyperedges $E_1,E_2$ containing $x$ satisfy $E_j\cap S=\{x\}$. Let $I\subseteq S$ be the set of isolated variables. Notice that any non-isolated variable in $S$ must be in at least one hyperedge with at least two variables in $S\setminus I$. Also note that any hyperedge can contain at most one isolated variable.

    \textbf{Claim:} $|I|\leq k-|S|$. Consider the bipartite graph $R\subseteq I\times(V\setminus S)$ where $(x,u)\in R$ iff $u\in E_1\cup E_2$, where $E_1,E_2$ are the two hyperedges containing $x$. Each $x\in I$ has at least two unique neighbors in $(V\setminus S)$: neither $E_1$ nor $E_2$ is unary, and they share at most one variable ($x$), so each contributes a unique vertex in $V\setminus S$. Each $u\in V\setminus S$ has at most two neighbors in $I$, since it only appears in two hyperedges, each of which has at most one isolated variable.
    
    Now for any $X\subseteq I$, there are at least $2|X|$ edges of $R$ incident to $X$: each element of $X$ contributes at least two edges, one from each of its two hyperedges. On the other hand, a subset $N \subseteq (V\setminus S)$ has at most $2|N|$ edges of $R$ incident to it: each element of $N$ can have at most two neighbors in $I$. Therefore, the set $N(X) \subseteq V\setminus S$ of vertices connected to $X$ satisfies $|N(X)|\geq|X|$. By Hall's condition on bipartite matching, we can match every element of $I$ to a distinct element of $V\setminus S$.
    
    This means that we can assign weight $1$ to each variable in $I$, and amortize a cost $1/2$ to a unique variable in $(V\setminus S)$. We then assign weight $1/2$ to the remaining variables in $S\setminus I$, which covers remaining hyperedges since they all have at least two variables in $S\setminus I$. This proves the lemma.
\end{proof}

\paragraph{Grid queries}

\begin{definition}
    A \textit{grid query} with parameters $r,c \in \mathbb{N}$ is a query with $r \times c$ variables, organized in an $r \times c$ grid. We will denote the variables as $x_{i,j}$ for $1 \leq i \leq r$ and $1 \leq j \leq c$. The query has $r+c$ edges, one for each row and each column of variables: $\{x_{i,1}, \dots, x_{i,c}\}$ for each row $i$, and $\{x_{1,j}, \dots, x_{r,j}\}$ for each column $j$.
\end{definition}

\begin{lemma}
\label{lemma:gridkappa}
    For an $r \times c$ grid query with $r,c\geq 2$, $\kappa \leq \max\{\max\{r,c\}, r+c-4\}$.
\end{lemma}
\begin{proof}
    We show $\kappa \leq \max\{\max\{r,c\},\, r+c-4\}$ by constructing a cover for any $S \subseteq V$.

    Let $s$ be the number of row-column pairs $i,j$ such that row $i$ only contains one variable $j$ in $S$, and column $j$ only contains one variable $i$ in $S$. These pairs give singleton edges: each gets weight $1$ in the vertex cover. We can ignore these rows and columns, since proving the bound on remaining rows and columns gives the same bound (we remove one row and column, and pay $1$ weight). By a similar argument, we can ignore the case where $x_{i,j}$ is the only variable in $S$ in row $i$ but column $j$ has $\geq 2$ variables in $S$, or the opposite, since we can put weight $0$ on $x_{i,j}$: the corresponding hyperedge gets removed by reduction.

    Therefore, focus on the case where each row and column has at least $2$ variables in $S$.

    \textbf{Case 1: $\min\{r,c\} < 4$.} W.l.o.g.\ $r \leq c$. We will put weight $1$ in one variable per column, with $0$ weight on the other variables and show the picks can cover all $r \leq 3$ rows. Row 1 takes any of its $\geq 2$ variables. Row~2 (if it exists) has $\geq 2$ variables, and only one can equal row~1's choice of column, so a new column can be covered. For row~3, it again has $\geq 2$ variables. If it is in a non-covered column, use it. Otherwise, we do a swap: give row~3 a column $j$ it has a variable in, displacing the row that held $j$. If the displaced row has a free neighbor, it takes it instead and we are done. Otherwise, for every neighbor $j$ of row~3, the row holding $j$ has all its neighbors among the (at most $2$) taken columns. But then all three rows connect only to these $\leq 2$ columns, making $c \leq 2 < 3 \leq r$, contradicting $r \leq c$.

    \textbf{Case 2: $\min\{r,c\} \geq 4$.} We will find 4 variables $x_{i,j}$ which are pairwise distinct in both row and column. We will create the vertex cover by assigning weight~1 to each such pair, plus one variable per remaining row-edge and one per remaining column-edge. This establishes the total cost as: $4 + (r-4) + (c-4)$ as we want.

    Take a maximum collection $M$ of pairs from $S$ with pairwise distinct rows and pairwise distinct columns. Assume $|M| \leq 3$ for contradiction.

    Since $r, c \geq 4 > |M|$, there is a row $i^*$ and a column $j^*$ not used by $M$. Every variable of $i^*$ in $S$ lies in a column used by $M$ (otherwise we can add the pair to $M$). Since $i^*$ has $\geq 2$ such variables, at least $2$ distinct columns of $M$ contain a variable of $i^*$ in $S$. The same for $j^*$: at least $2$ distinct rows of $M$ contain a variable of $j^*$ in $S$.

    Now consider any column $j$ of $M$ with $x_{i^*, j} \in S$, paired in $M$ with a row $r$. If also $x_{r, j^*} \in S$, then $M \setminus \{x_{r,j}\} \cup \{x_{i^*,j},x_{r,j^*}\}$ is a larger valid collection, contradicting maximality. So $(r, j^*) \notin S$.

    Call a row of $M$ \emph{red} if it is paired (in $M$) with a column containing a variable of $i^*$ in $S$, and \emph{blue} if it itself contains a variable of $j^*$ in $S$. By the previous paragraph, there are $\geq 2$ red rows and $\geq 2$ blue rows, and no row is both red and blue. Hence $M$ has $\geq 4$ pairs, contradicting $|M| \leq 3$.
\end{proof}

Together with the matching lower bound from \cite{pacv2}, this shows that $\kappa$ characterizes the optimal load for grid queries.

\iffull
\section{Lower Bound}
\label{sec:lb:appendix}

\begin{lemma}\label{lemma:OptimalPartitionLPIsVC}
    For a query $Q$ with hypergraph $\mH=(V,E)$, let $\ve{v}^*$ be an optimal solution to the following linear program.
    \begin{equation}\label{LP:LPSame1}
        \begin{aligned}
            \min_{\ve{v}} \quad & \sum_{x\in V}v_x  && \\
            \text{s.t.}   \quad & \sum_{x\in R}v_x \geq 1, && \forall R\in Q \\
                                & v_x\geq 0,               && \forall x\in V
        \end{aligned}
    \end{equation}
    Consider now the following linear program
    \begin{equation}\label{LP:LPSame2}
        \begin{aligned}
            \max_\psi   \quad & \sum_{x\in V}\psi_x && \\
            \text{s.t.} \quad & \sum_{x\in R}\psi_x\leq1, && \forall R\in Q \text{ where } \sum_{x\in R}v^*_x=1 \\
                              & \psi_x \leq 0,            && \forall x\in V \text{ where } v^*_x=0 \\
                              & \psi_x\geq 0,             && \forall x\in V
        \end{aligned}
    \end{equation}
    Then, the optimal solution $\ve{v}^*$ in the first LP is also optimal in the second LP.
\end{lemma}
\begin{proof}
    $\ve{v}^*$ is feasible in LP~\ref{LP:LPSame2} since LP~\ref{LP:LPSame2} does not have a constraint $\sum_{x\in R}\psi_x\leq 1$ for $R$ where $\sum_{x\in R}v_x>1$.
    
    Suppose for contradiction that $\ve{v}^*$ is not optimal in LP ~\ref{LP:LPSame2}. Then, there exists a vector $\Delta$ such that $\ve{v}^*+\Delta$ is a better solution.
    
    We claim that then, for some $\epsilon >0$, $\ve{v}^*-\epsilon\Delta$ is a solution to LP~\ref{LP:LPSame1}. To show this, we will argue that all constraints are satisfied. First, for any $R$ where LP~\ref{LP:LPSame1} has a tight constraint, $\sum_{x\in R}\Delta_x\leq0$, so the corresponding constraint in LP~\ref{LP:LPSame1} is satisfied. For any $R$ such that the constraint in LP~\ref{LP:LPSame1} is slack, then there exists an $\epsilon$ such that $\sum_{x\in R}(v_x-\epsilon\Delta_x)\geq 1$.
    
    Lastly, we argue that the constraints $v_x\geq 0$ are satisfied, even if $\Delta_x\geq0$. We know that if $v_x=0$, then $\Delta_x=0$. If $v_x>0$, then we can pick $\epsilon$ sufficiently small such that $v_x-\epsilon\Delta_x\geq0$.

    Notice that $\ve{v}^*-\epsilon\Delta$ is a vertex cover of lower weight than $\ve{v}^*$. This is a contradiction, since $\ve{v}^*$ was picked because it is the smallest vertex cover.
\end{proof}


\begin{lemma}
    Let $Q$ be a sparse product query with a reduced hypergraph $\mH$.
    Consider a machine which is permitted to load $L$ tuples from each relation. 
    Then, in expectation (over the randomness of the sparse relations) the  number of tuples the machine will produce is at most 
    \[
        L^{\tau^*(\mH)} \prod_{R \in Q}P_R
    \]
\end{lemma}
\begin{proof}
    For each of the dense relations, we can load at most $L$ tuples.
    This means that the maximal number of output tuples that can be produced is given by
    \begin{align*}
        \max_\Psi   \quad & \prod_{x\in V}\Psi_x  && \\
        \text{s.t.} \quad & \prod_{x\in R}\Psi_x\leq L, && \forall R\in Q \text{ where } \sum_{x\in R}v^*_x=1 \\
                          & \Psi_x \leq n^{v^*_x},      && \forall x\in V
    \end{align*}

    By taking the logarithm with base $L$, we get the following linear program, which is a vertex packing problem
    \begin{align*}
        \max_\psi   \quad & \sum_{x\in V}\psi_x  && \\
        \text{s.t.} \quad & \sum_{x\in R}\psi_x\leq1,  && \forall R\in Q \text{ where } \sum_{x\in R}v^*_x=1 \\
                          & \psi_x \leq v^*_x\log_Ln, && \forall x\in V
    \end{align*}

    Note that $\log_L n > 1$ since $L<n$. Instead of the last constraint, we will have the constraint $\forall x|v_x=0: \psi_x\leq 0$. The LP obtained by changing the constraints in this way keeps all feasible solutions feasible, but makes some infeasible points feasible. This means that if we can prove that the new LP satisfies the bound, so does the original LP.
    
    We now apply Lemma~\ref{lemma:OptimalPartitionLPIsVC}, which gives that the above optimization problem has optimal solution $L^{\tau^*}$. We obtain the desired result by summing over the expectation that any tuple formed by the dense relations survives to the output, which is $\Pi_{R} P_R$.
\end{proof}




\fi

\end{document}